\documentclass[letterpaper,11pt]{article}
\usepackage{mathrsfs}
\usepackage[mathscr]{euscript}

\newcommand{\longversion}[1]{#1}
\newcommand{\shortversion}[1]{}

\shortversion{
 \usepackage[font=footnotesize]{caption}
 \newcommand{\mymath}[1]{$#1$}

}
\longversion{
 \newcommand{\mymath}[1]{\begin{align*}#1\end{align*}}

}

\usepackage{amsfonts,amsthm,amsmath}
\usepackage{amstext}
\usepackage{graphicx,tikz}
\RequirePackage{fancyhdr}
\usepackage{xcolor}
\usepackage{boxedminipage}

\makeatletter
\newenvironment{subxarray}{%
  \vcenter\bgroup
  \Let@ \restore@math@cr \default@tag
  \baselineskip\fontdimen10 \scriptfont\tw@
  \advance\baselineskip\fontdimen12 \scriptfont\tw@
  \lineskip\thr@@\fontdimen8 \scriptfont\thr@@
  \lineskiplimit\lineskip
  \ialign\bgroup\hfil
    $\m@th\scriptstyle##$&$\m@th\scriptstyle{}##$\hfil\crcr
}{%
  \crcr\egroup\egroup
}
\makeatother

\usepackage{listings}

\usepackage[margin=1in]{geometry}

\newcommand{\yes}{{yes}}
\newcommand{\no}{{no}}
\newcommand{\yesinstance}{\yes-instance\xspace}

\newcommand{\bran}[1]{branchable\xspace}

\usepackage[linesnumbered, boxed, algosection,noline]{algorithm2e}

\newtheorem{theorem}{Theorem}
\newtheorem{lemma}{Lemma}[section]

\newtheorem{definition}{Definition}[section]

\newtheorem{proposition}{Proposition}[section]

\usepackage{amsthm}
\usepackage{thmtools}
\declaretheoremstyle[headfont=\bf]{normalhead}
\declaretheorem[style=normalhead]{remark}

\date{}

\newcommand{\defparproblem}[4]{
  \longversion{\medskip}
  \shortversion{\smallskip}
\noindent\fbox{
  \begin{minipage}{0.96\textwidth}
  \begin{tabular*}{\textwidth}{@{\extracolsep{\fill}}lr} #1  & {\bf{Parameter:}} #3 \\ \end{tabular*}
  {\bf{Input:}} #2  \\
  {\bf{Question:}} #4
  \end{minipage}
  }
  \longversion{\medskip}
  \shortversion{\smallskip}
}

\newcommand{\defproblem}[3]{
	\longversion{\medskip}
	\shortversion{\smallskip}
\noindent\fbox{
  \begin{minipage}{0.96\textwidth}
  \begin{tabular*}{\textwidth}{@{\extracolsep{\fill}}lr} #1 \\ \end{tabular*}
  {\bf{Input:}} #2  \\
  {\bf{Question:}} #3
  \end{minipage}
  }
  \longversion{\medskip}
  \shortversion{\smallskip}
}

\newcommand{\defoptproblem}[3]{
	\longversion{\medskip}
	\shortversion{\smallskip}
	\noindent\fbox{
		\begin{minipage}{0.96\textwidth}
			\begin{tabular*}{\textwidth}{@{\extracolsep{\fill}}lr} #1 \\ \end{tabular*}
			{\bf{Input:}} #2  \\
			{\bf{Output:}} #3
		\end{minipage}
	}
	\longversion{\medskip}
	\shortversion{\smallskip}
}

\usepackage{graphics}

\usepackage{pgf}
\usepackage{paralist} 
\usepackage{graphicx}
\usepackage{boxedminipage}
\usepackage{boxedminipage}
\usepackage{xcolor}
\usepackage[linesnumbered,boxed,algosection]{algorithm2e}
\usepackage{framed}
\usepackage{booktabs}

\newcommand{\NP}{\text{\normalfont  NP}}

\newcommand{\FPT}{\text{\sf FPT}}

\newcommand{\W}[1][xxxx]{\text{\normalfont W[#1]}}

\usepackage{amsmath, amssymb, latexsym}
\usepackage{enumerate}

\usepackage{float}
\usepackage{xspace}

\newcommand{\fdel}{$\Pi$-\textsc{Vertex Deletion}}

\newcommand{\abGraph}{$(r,\ell)$-{graph}}
\newcommand{\abpartization}{{\sc Vertex $(r,\ell)$-Partization}}

\newcommand{\phisub}{{\sc  $\Phi$-Sub\-set}}
\newcommand{\dphisub}{\textsc{De\-ci\-sion}  \phisub}
\newcommand{\phiext}{{\sc  $\Phi$-Ex\-ten\-sion}}
\newcommand{\pphiext}{\textsc{Per\-mis\-sive} \phiext}

\newcommand{\sepfam}[3]{$(#1,#2,#3)$-set-inclusion-family} 
 
\newcommand{\sepfamwPlural}{set-inclusion-families}

\newcommand{\Oh}{{\mathcal{O}}}

\usepackage{todonotes}

\usepackage{todonotes}
 \usepackage[ pdftex, plainpages = false, pdfpagelabels, 
                 bookmarks=false,
                 bookmarksopen = true,
                 bookmarksnumbered = true,
                 breaklinks = true,
                 linktocpage,
                 pagebackref,
                 colorlinks = true,  
                 linkcolor = blue,
                 urlcolor  = blue,
                 citecolor = red,
                 anchorcolor = green,
                 hyperindex = true,
                 hyperfigures
                 ]{hyperref}

\newcommand{\bitsize}{N}

\usepackage{thm-restate}

\usepackage{thmtools}

\begin{document}

\title{Exact Algorithms via Monotone Local Search}

\author{Fedor V.\ Fomin\thanks{University of Bergen, Norway. 
     \texttt{\{fomin|daniello\}@ii.uib.no}} 
\and Serge Gaspers\thanks{The University of New South Wales, Sydney, Australia. \texttt{sergeg@cse.unsw.edu.au}} 
\thanks{Data61 (formerly: NICTA), CSIRO, Sydney, Australia} \addtocounter{footnote}{-3}
\and Daniel Lokshtanov\footnotemark \addtocounter{footnote}{-1}
 \and  Saket Saurabh\footnotemark \addtocounter{footnote}{2}~\thanks{The Institute of Mathematical Sciences, Chennai, India. \texttt{saket@imsc.res.in}.
 }}

 \maketitle

\begin{abstract} 
We give a new general approach for designing exact exponential-time algorithms for {\em subset problems}. In a subset problem the input implicitly describes a family of sets over a universe of size $n$ and  the task is to determine whether the family contains at least one set. A typical example of a subset problem is \textsc{Weighted $d$-SAT}. Here, the input is a CNF-formula with clauses of size at most $d$, and an integer $W$. The universe is the set of variables and the variables have integer weights. The family contains all the subsets $S$ of variables such that the total weight of the variables in $S$ does not exceed $W$, and setting the variables in $S$ to 1 and the remaining variables to 0 satisfies the formula.
Our approach is based on ``monotone local search'', where the goal is to extend a partial solution to a solution by adding as few elements as possible. More formally, in the extension problem we are also given as input a subset $X$ of the universe and an integer $k$. The task is to determine whether one can add at most $k$ elements to $X$ to obtain a set in the (implicitly defined) family. Our main result is that a  $c^kn^{\Oh(1)}$ time algorithm for the extension problem immediately yields a randomized algorithm for finding a solution of any size with running time $\Oh((2-\frac{1}{c})^n)$. 

In many cases, the extension problem can be reduced to simply finding a solution of size at most $k$. Furthermore, efficient algorithms for finding small solutions have been extensively studied in the field of parameterized algorithms. Directly applying these algorithms, our theorem yields in one stroke significant improvements over the best known exponential-time algorithms for several well-studied problems, including \textsc{$d$-Hitting Set},  \textsc{Feedback Vertex Set}, \textsc{Node Unique Label Cover}, and \textsc{Weighted $d$-SAT}. Our results demonstrate an interesting and very concrete connection between parameterized algorithms and exact exponential-time algorithms.

We also show how to derandomize our algorithms at the cost of a subexponential multiplicative factor in the running time. Our derandomization is based on an efficient construction of a new pseudo-random object that might be of independent interest. Finally, we extend our methods to establish new combinatorial upper bounds and develop enumeration algorithms.  

\end{abstract}

\thispagestyle{empty}

\newpage
\pagestyle{plain}
\setcounter{page}{1}

\section{Introduction}

In the area of exact exponential-time algorithms, the objective is to design algorithms that outperform brute-force for computationally intractable problems. Because the problems are intractable we do not hope for polynomial time algorithms. Instead the aim is to allow super-polynomial time and design algorithms that are significantly faster than brute-force. For {\em subset problems} in \NP, where the goal is to find a subset with some specific properties in a universe on $n$ elements, the brute-force algorithm that tries all possible solutions has running time $2^nn^{\Oh(1)}$. Thus our goal is typically to design an algorithm with running time $c^nn^{\Oh(1)}$ for $c < 2$, and we try to minimize the constant $c$. We refer to the textbook of Fomin and Kratsch~\cite{Fomin:2010mo} for an introduction to the field.

%
%
 %
%
%
In the area of parameterized algorithms (see~\cite{cygan2015parameterized}), the goal is to design efficient algorithms for the ``easy'' instances of computationally intractable problems. Here the running time is measured not only in terms of the input size $n$, but also in terms of a parameter $k$ which is expected to be small for ``easy'' instances. For subset problems the parameter $k$ is often chosen to be the size of the solution sought for, and many subset problems have parameterized algorithms that find a solution of size $k$ (if there is one) in time $c^kn^{\Oh(1)}$ for a constant $c$, which is often much larger than $2$.


In this paper we address the following question: Can an {\em efficient} algorithm for the {\em easy} instances of a hard problem yield a {\em non-trivial} algorithm for {\em all} instances of that problem? More concretely, can parameterized algorithms for a problem be used to speed up exact exponential-time algorithms for the same problem? Our main result is an affirmative answer to this question: we show that, for a large class of problems, an algorithm with running time $c^kn^{\Oh(1)}$ for any $c > 1$ immediately implies an exact  algorithm with running time $\Oh((2-\frac{1}{c})^{n+o(n)})$ for the problem. Our main result, coupled with the fastest known parameterized algorithms, gives in one stroke the first non-trivial exact algorithm for a number of problems, and simultaneously significantly improves over the best known exact algorithms for several well studied problems; see Table~\ref{fig:vertexresults} for a non-exhaustive list of corollaries. Our approach is also useful to prove upper bounds on the number of interesting combinatorial objects, and to design efficient algorithms that enumerate these objects; see Table~\ref{fig:enumresults}.


At this point it is worth noting that a simple connection between algorithms running in time $c^kn^{\Oh(1)}$ for $c < 4$ and exact exponential-time algorithms beating $\Oh(2^n)$ has been known for a long time.
%
For subset problems, where we are looking for a specific subset of size $k$ in a universe of size $n$, to beat $\Oh(2^n)$ one only needs to outperform brute-force for values of $k$ that are very close to $n/2$. Indeed, for $k$ sufficiently far away from $n/2$, trying all subsets of size $k$ takes time ${n \choose k}n^{\Oh(1)}$ which is significantly faster than $\Oh(2^n)$.
%
%
%
Thus, if there is an algorithm deciding whether there is a solution of  size at most $k$ in time $c^k n^{\Oh(1)}$ for some $c<4$, we can deduce that the problem can be solved in time $\Oh((2-\epsilon)^n)$ for an $\epsilon > 0$ that depends only on $c$. On the other hand, it is easy to see that this trade-off between $c^k$ and ${n \choose k}$ does not yield any improvement over $2^n$ when $c \geq 4$. Our main result significantly outperforms the algorithms obtained by this trade-off for every value of $c > 1$, and further yields better than $\Oh(2^n)$ time algorithms even for $c \geq 4$.

As a concrete example, consider \longversion{the }{\sc Interval Vertex Deletion}\longversion{ problem}. Here the input is a graph $G$ and an integer $k$ and the task is to determine whether $G$ can be turned into an interval graph by deleting $k$ vertices. The fastest parameterized algorithm for the problem \longversion{is due to Cao  }\cite{Cao8kinterval}\longversion{ and} runs in time $8^k n^{\Oh(1)}$. Combining this algorithm with the simple trade-off scheme described above does not outperform brute-force, since $8\ge 4$. The fastest previously known exponential-time algorithm \longversion{for the problem is due to Bliznets et al.~}\cite{BliznetsFPV13}\longversion{, and} runs in time $\Oh((2-\epsilon)^n)$ for $\epsilon < 10^{-20}$. On the other hand, combining the parameterized algorithm, as a black box, with our main result immediately yields a $1.875^{n+o(n)}$ time algorithm for \longversion{{\sc Interval Vertex Deletion}}\shortversion{the problem}.

\medskip\noindent\textbf{Our results.}
We need some definitions in order to state our results precisely. 
We define an {\em implicit set system} as a function $\Phi$ that takes as input a string $I \in \{0,1\}^*$ and outputs a set system $(U_I, {\cal F}_I)$, where $U_I$ is a universe  and ${\cal F}_I$ is a collection of subsets of $U_I$.  The string $I$ is referred to as an {\em instance}  and we denote by $|U_I| = n$ the size of the universe and by $|I|=\bitsize$ the size of the instance. 
We assume that $\bitsize\ge n$. The implicit set system $\Phi$ is said to be {\em polynomial time computable} if (a) there exists a polynomial time algorithm that given $I$ produces $U_I$, and (b) there exists a polynomial time algorithm that given $I$, $U_I$ and a subset $S$ of $U_I$ determines whether $S \in{\cal F}_I$.
All implicit set systems discussed in this paper are polynomial time computable, except for the minimal satisfying assignments of $d$-CNF formulas which are not polynomial time computable unless P=NP \cite{YatoS03}.

An implicit set system $\Phi$ naturally leads to some problems about\longversion{ the set system} $(U_I, {\cal F}_I)$. Find a set in $\mathcal{F}_I$. Is $\mathcal{F}_I$ non-empty? What is the cardinality of $\mathcal{F}_I$? In this paper we will primarily focus on the first and last problems. Examples of implicit sets systems include the set of all feedback vertex sets of a graph of size at most $k$, the set of all satisfying assignments of a CNF formula of weight at most $W$, and the set of all minimal hitting sets of a set system. Next we formally define subset problems.

\defoptproblem{\phisub}{An instance $I$}{A set $S \in {\cal F}_I$ if one exists.}

\noindent
An example of a subset problem is \textsc{Min-Ones} $d$-\textsc{Sat}. Here for an integer $k$ and a propositional formula $F$ in conjunctive normal form (CNF) where each clause has at most $d$ literals,
the task is to determine whether $F$ has a satisfying assignment with Hamming weight \shortversion{(number of $1$s) }at most $k$\longversion{, i.e., setting at most $k$ variables to 1}. 
In our setting, the instance $I$ consists of the input formula $F$ and the integer $k$, encoded as a string over $0$s and  $1$s. The implicit set system  $\Phi$ is a function from $I$ to $(U_I, {\cal F}_I)$, where $U_I$ is the set of variables of $F$, and ${\cal F}_I$ is the set of all satisfying assignments of Hamming weight at most $k$.
 
%
%

Our results will rely on parameterized algorithms for a generalization of subset problems\shortversion{, called {\em extension problems},} where we are also given as input a set $X \subseteq U_I$ and an integer $k$ and the question is whether it is possible to obtain a set from $\mathcal{F}_I$ by adding at most $k$ elements from $U_I$ to $X$.\longversion{ We give a formal definition of such problems, which we call {\em extension problems}.}

\defproblem{\phiext}{An instance $I$, a set $X \subseteq U_I$, and an integer $k$.}
{Does there exists a subset $S \subseteq (U_I \setminus X)$ such that $S \cup X \in {\cal F}_I$ and $|S| \leq k$?}

\noindent
Our first main result gives exponential-time randomized algorithms for \phisub\ based on single-exponential parameterized algorithms for \phiext\ with parameter $k$. Our randomized algorithms are Monte Carlo algorithms with one-sided error. On \no{}-instances they always return \no{}, and on \yes-instances they return \yes\ (or output a certificate) with probability $>\frac{1}{2}$.

\begin{restatable}{theorem}{thmMainOne}
	\label{thm:main1}
	If there exists an algorithm for \phiext{} with running time $c^k\bitsize^{\Oh(1)}$ then there exists a randomized algorithm for \phisub\ with running time $(2-\frac{1}{c})^{n}\bitsize^{\Oh(1)}$.
\end{restatable}

\noindent
Our second main result is that the algorithm of Theorem~\ref{thm:main1} can be derandomized at the cost of a subexponential factor in $n$ in the running time.

\begin{restatable}{theorem}{thmMainTwo}
	\label{thm:main2}
	If there exists an algorithm for \phiext{} with running time $c^k\bitsize^{\Oh(1)}$ then there exists an algorithm for \phisub\ with running time $(2-\frac{1}{c})^{n+o(n)} \bitsize^{\Oh(1)}$.
\end{restatable}

\noindent
To exemplify the power of these theorems, we give a few examples of applications. We have already seen the first example, the $1.875^{n+o(n)}$ time algorithm for {\sc Interval Vertex Deletion}. Let us now consider the \textsc{Min-Ones $d$-SAT} problem described above. 

A simple branching algorithm solves the extension problem for \textsc{Min-Ones $d$-SAT} as follows. Suppose we have already selected a set $X$ of variables to set to 1, remove all clauses containing a positive literal on $X$, and remove negative literals on $X$ from the remaining clauses. Start from the all-$0$ assignment on the remaining variables, with a budget for flipping $k$ variables from 0 to 1. As long as there is an unsatisfied clause, guess which one of the at most $d$ variables in this clause should be flipped from 0 to 1, and for each proceed recursively with the budget decreased by one. The recursion tree of this algorithm has depth at most $k$, and each node of the recursion tree has at most $d$ children, thus this algorithm terminates in time $d^k\cdot n^{\Oh(1)}$.

Hence, by Theorem~\ref{thm:main2},  \textsc{Min-Ones $d$-SAT} can be solved in time  $(2-\frac{1}{d})^{n+o(n)}$. For $d=3$ there is a faster parameterized algorithm with running time $2.562^k n^{\Oh(1)}$ \longversion{due to  Kutzkov and  Scheder  }\cite{abs-1007-1166}. Thus \textsc{Min-Ones $3$-SAT} can be solved in time $\Oh(1.6097^n)$. 
Since  \textsc{$d$-Hitting Set} is a special case of \textsc{Min-Ones $d$-SAT}, the same bounds hold for this problem as well, and the same approach works for weighted variants of these problems.
However, due to faster known   parameterized algorithms for  \textsc{$d$-Hitting Set}, our theorem implies faster exact algorithms for \textsc{$d$-Hitting Set} with running time $(2-\frac{1}{(d-0.9255)})^n$.

Another interesting example is \longversion{the }\textsc{Feedback Vertex Set}\longversion{ problem}. Here the task is to decide, for a graph $G$ and an integer $k$, whether $G$ can be made acyclic by removing $k$ vertices.
While this problem is trivially solvable in time  $2^n n^{\Oh(1)}$ for $n$-vertex graphs,
breaking the $2^n$-barrier  for the problem was an open problem \longversion{in the area }for some time.
The first algorithm breaking the barrier is due
to Razgon~\cite{Razgon06}.
\shortversion{After a series of improvements \cite{FominGPR08-On,XiaoN13,FominTV15}, an  $\Oh(1.7347^n)$ time algorithm was obtained in \cite{FominTV15}.}%
\longversion{The running time $\Oh(1.8899^n)$ of the algorithm from~\cite{Razgon06} was improved in~\cite{FominGPR08-On} to
$\Oh(1.7548^n)$. Then Xiao and Nagamochi~\cite{XiaoN13} gave an algorithm with running time   $\Oh(1.7356^n)$. 
Finally an algorithm of running time $\Oh(1.7347^n)$ was obtained in \cite{FominTV15}.}
For the parameterized version of the problem there was also a chain of improvements \cite{CaoCL15,ChenFLLV08,DehneFLRS07,GuoGHNW06} resulting in   a $3^k  n^{\Oh(1)}$ time randomized algorithm   \cite{cut-and-count}
and a $3.591^k n^{\Oh(1)}$ time deterministic algorithm  \cite{KociumakaP13}. This, coupled with our main theorem, immediately gives us randomized and deterministic algorithms of running times $\Oh(1.6667^n)$ and $\Oh(1.7216^n)$, respectively.

More generally, let $\Pi$ be a hereditary family of graphs. That is, if $G \in \Pi$ then so are all its induced subgraphs. Examples of hereditary families include the edgeless graphs, forests, bipartite graphs, chordal graphs, interval graphs, split graphs, and perfect graphs.  Of course this list is not exhaustive. For every hereditary graph family $\Pi$ there is a natural vertex deletion problem\shortversion{:}\longversion{, that we define here.}

\defproblem{\fdel}{An undirected (or directed) graph $G$ and an integer $k$.}{Is there a set $S\subseteq V(G)$ with $|S|\le k$ such that $G-S\in \Pi$?}
\smallskip

\noindent 
We can cast \fdel\ as a \phisub\ problem as follows. The instance $I$ describes the graph $G$, so $U_I = V(G)$ and ${\cal F}_I$ is the family of subsets $S$ of $V(G)$ of size at most $k$ such that $G - S \in \Pi$. Notice that a polynomial time algorithm to determine whether a graph $G$ is in $\Pi$ yields a polynomial time algorithm to determine whether a set $S$ is in ${\cal F}_I$, implying that $\Phi$ is polynomial time computable.
Moreover, a $c^k \bitsize^{\Oh(1)}$ time algorithm for \fdel\ trivially gives the same running time for its extension variant, since vertices in $X$ can simply be deleted.
Also, if $\Pi$ is characterized by a set of forbidden induced subgraphs which all have at most $d$ vertices, such as cographs ($d=4$) and split graphs ($d=5$), we can reduce the \fdel\ problem to $d$-Hitting Set where the number of elements is the number of vertices of the input graph.

\newcommand{\tcite}[1]{\hfill\cite{#1}}

\begin{table}[t]
	\centering
	\setlength{\tabcolsep}{5pt}
	{\footnotesize
		\begin{tabular}{l l l l}
			\toprule
			Problem Name       &  Parameterized     &  New bound                                & Previous Bound       \\
			\midrule
			{\sc Feedback Vertex Set}   & $3^k$ (r) \tcite{cut-and-count} &      $1.6667^n$   (r)                  &    \\
			{\sc Feedback Vertex Set} &  $3.592^k$            \tcite{KociumakaP13}          &  $1.7217^n$    & $1.7347^n$ \tcite{FominTV15}  \\
			{\sc Subset Feedback Vertex Set} &            $4^k$         \tcite{Wahlstrom14}        &   $1.7500^n$   &  $1.8638^n$ \tcite{FominHKPV14}  \\
			{\sc Feedback Vertex Set in Tournaments} &     $1.6181^k$        \tcite{KumarL16}                 &   $1.3820^n$   &  $1.4656^n$  \tcite{KumarL16}  \\
			{\sc  Group Feedback Vertex Set} &          $4^k$           \tcite{Wahlstrom14}                 &   $1.7500^n$   & NPR    \\
			\textsc{Node Unique Label Cover} &          $|\Sigma|^{2k}$           \tcite{Wahlstrom14}                 &   $(2-\frac{1}{|\Sigma|^2})^n$   & NPR    \\
			{\abpartization} ($r,\ell \leq 2$) &   $3.3146^k$   \tcite{BasteFKS15,KolayP15}  & $1.6984^n$ &  NPR  \\ 
			{\sc Interval Vertex Deletion}&                   $8^k$       \tcite{Cao8kinterval}   &  $1.8750^n$    & $(2-\varepsilon)^n$ for $\varepsilon <10^{-20}$  \tcite{BliznetsFPV13} \\
			{\sc Proper Interval Vertex Deletion} &     $6^k$           \tcite{HofV13,Cao15}           &   $1.8334^n$   &  $(2-\varepsilon)^n$ for $\varepsilon <10^{-20}$  \tcite{BliznetsFPV13} \\
			{\sc Block Graph Vertex Deletion}&                   $4^k$   \tcite{AgrawalLKS16}      &    $1.7500^n$   & $(2-\varepsilon)^n$ for $\varepsilon <10^{-20}$  \tcite{BliznetsFPV13}  \\
			{\sc   Cluster Vertex Deletion}&                   $1.9102^k$      \tcite{BoralCKP14}    &   $1.4765^n$   &  $1.6181^n$  \tcite{FominGKLS10}  \\
			{\sc   Thread Graph Vertex Deletion} &                   $8^k$    \tcite{Kante0KP15}  &    $1.8750^n$    & NPR    \\
			{\sc   Multicut on Trees} &                   $1.5538^k$  \tcite{KanjLLTXXYZZZ14}   &    $1.3565^n$    & NPR    \\
			{\sc    $3$-Hitting Set} &                   $2.0755^k$    \tcite{MagnusPhD07}    &  $1.5182^n$    &   $1.6278^n$    \tcite{MagnusPhD07}  \\
			{\sc   $4$-Hitting Set}&                   $3.0755^k$      \tcite{FominGKLS10}      &  $1.6750^n$    &   $1.8704^n$ \tcite{FominGKLS10}     \\
			{\sc   $d$-Hitting Set} ($d\geq 3$)&                   $(d-0.9245)^k$     \tcite{FominGKLS10}        &  $(2-\frac{1}{(d-0.9245)})^n$    & \tcite{CochefertCGK16,FominGKLS10}   \\
			{\sc    Min-Ones $3$-SAT} &                   $2.562^k$    \tcite{abs-1007-1166}      &   $1.6097^n$   &  NPR   \\
			{\sc    Min-Ones $d$-SAT} ($d\geq 4$) &                   $d^k$          &  $(2-\frac{1}{d})^n$    &   NPR    \\
			{\sc    Weighted $d$-SAT} ($d\geq 3$) &                   $d^k$          &  $(2-\frac{1}{d})^n$    &   NPR    \\
			{\sc Weighted Feedback Vertex Set} & $3.6181^k$   \tcite{AgrawalLKS16}                 &   $1.7237^n$     &  $1.8638^n$  \tcite{FominGPR08-On} \\
			\textsc{Weighted 3-Hitting Set} & $2.168^k$ \tcite{ShachnaiZ15} & $1.5388^n$ & $1.6755^n$ \tcite{CochefertCGK16}\\
			\textsc{Weighted $d$-Hitting Set} ($d\ge 4$) & $(d-0.832)^k$ \tcite{FominGKLS10,ShachnaiZ15} & $(2-\frac{1}{d-0.932})^n$ & \tcite{CochefertCGK16} \\
			\bottomrule
		\end{tabular}
		
	}
	\caption{\label{fig:vertexresults}Summary of known and new results for different 
		optimization problems.
		NPR means  that we are not aware of any previous algorithms faster than brute-force. All bounds suppress factors polynomial in the input size $N$.\longversion{ The algorithms in the first row are randomized (r).}}
\end{table}


In Table~\ref{fig:vertexresults} we list more applications of Theorem~\ref{thm:main2}.
We also provide the running times of the fastest known parameterized and exact algorithms. The problem definitions are given in \longversion{the appendix}\shortversion{Appendix \ref{sec:problems}}.
For most \longversion{of these }problems, the results are obtained by simply using the fastest known parameterized algorithm in combination with Theorem \ref{thm:main2}.
\longversion{Some of the results (for \textsc{Weighted} $d$-\textsc{Hitting Set})}\shortversion{The results for \textsc{Weighted} $d$-\textsc{Hitting Set}} follow from a variant that relies on algorithms for a permissive version of \phiext; see \longversion{Subsection \ref{subsec:permissive}}\shortversion{Appendix \ref{subsec:permissive}}.

We also extend the technique \longversion{developed for decision problems }to enumeration problems and \longversion{to prove }combinatorial upper bounds. For example, a minimal satisfying assignment of a $d$-CNF formula is a satisfying assignment $a$ such that no other satisfying assignment sets every variable to $0$ that $a$ sets to $0$.
It is interesting to investigate the number of minimal satisfying assignments of $d$-CNF formulas, algorithms to enumerate these assignments, and upper bounds and enumeration algorithms for other combinatorial objects.


Formally, let $\Phi$ be an implicit set system and $c \geq 1$ be a real-valued constant. We say that $\Phi$ is {\em $c$-uniform} if, for every instance $I$, set $X \subseteq U_I$, and integer $k \leq n - |X|$, the cardinality of the collection
\mymath{{\mathcal F}_{I,X}^k =\left\{S \subseteq U_I \setminus X~:~|S| = k \text{ and } S \cup X \in {\cal F}_I \right\}}
is at most $c^k n^{\Oh(1)}$. The next theorem will provide new \longversion{combinatorial }upper bounds for collections generated by $c$-uniform implicit set systems.

\begin{restatable}{theorem}{thmComb}
	\label{thm:main3}
	Let $c > 1$ and $\Phi$ be an implicit set system. If $\Phi$ is $c$-uniform, then $|{\cal F}_I| \leq \left(2-\frac{1}{c}\right)^n n^{\Oh(1)}$ for every instance $I$. 
\end{restatable}

\noindent
We say that an implicit set system $\Phi$ is {\em efficiently $c$-uniform} if there exists an algorithm that given $I$, $X$ and $k$ enumerates all elements of  ${\mathcal F}_{I,X}^k$ in time $c^k \bitsize^{\Oh(1)}$.
In this case, we can enumerate $\mathcal{F}_I$ in the same time, up to a subexponential factor in $n$.

\begin{restatable}{theorem}{thmEnum}
	\label{thm:main4}
	Let $c > 1$ and $\Phi$ be an implicit set system. If $\Phi$ is efficiently $c$-uniform, then there is an algorithm that given as input $I$ enumerates ${\cal F}_I$ in time $\left(2-\frac{1}{c}\right)^{n+o(n)} \bitsize^{\Oh(1)}$. 
\end{restatable}

\noindent
For minimal satisfying assignments of $d$-CNF formulas, we observe that the afore-mentioned branching algorithm for the extension version of \textsc{Min-Ones} $d$-\textsc{Sat}, which explores the Hamming ball of radius $k$ around the all-$0$ assignment of the reduced instance, encounters all minimal satisfying assignments extending $X$ by at most $k$ variables. Thus, minimal satisfying assignments for $d$-CNF formulas are $d$-uniform. It follows immediately that minimal $d$-hitting sets are $d$-uniform and they are also efficiently $d$-uniform.
\longversion{
	
}%
By a classical theorem of Moon from 1971 \cite{Moon:1971xz},
the number of maximal transitive subtournaments in an $n$-vertex tournament does not exceed 
 $1.7170^n$. In   \cite{GaspersM13}, Gaspers and Mnich improved this bound to $1.6740^n$. Our approach yields immediately  a better bound of $\Oh(1.6667^n)$
 since every directed 3-cycle needs to be hit.
 Similarly, in chordal graphs, a set is a feedback vertex set (FVS) if it hits every 3-cycle.
 For maximal $r$-colorable induced subgraphs of perfect graphs it suffices to hit every clique of size $r+1$. Some consequences of our results for enumeration algorithms and combinatorial bounds are given in Table~\ref{fig:enumresults}.

\begin{table}[t]
	\centering
	\setlength{\tabcolsep}{4pt}
	{\footnotesize
		\begin{tabular}{l l l l}
			\toprule
			Problem Name       &  $c$-uniform     &  New bound                                & Previous Bound       \\
			\midrule
			{Minimal FVSs in Tournaments} &     $3$                 &   $1.6667^n$   &  $1.6740^n$  \tcite{GaspersM13}  \\
			Minimal 3-{Hitting Sets} &          $3$     &   $1.6667^n$   & $1.6755^n$ \tcite{CochefertCGK16}    \\
			Minimal 4-{Hitting Sets} &          $4$     &   $1.7500^n$   & $1.8863^n$ \tcite{CochefertCGK16}    \\
			Minimal 5-{Hitting Sets} &          $5$     &   $1.8000^n$   & $1.9538^n$ \tcite{CochefertCGK16}    \\
			Minimal $d$-{Hitting Sets} &          $d$     &   $(2-\frac{1}{d})^n$   & $(2-\epsilon_d)^n$ with $\epsilon_d<1/d$ \tcite{CochefertCGK16}    \\
			Minimal $d$-{\sc Sat} ($d\ge 2$) &          $d$     &   $(2-\frac{1}{d})^n$   & NPR    \\
			Minimal FVSs in chordal graphs &          $3$     &   $1.6667^n$   & $1.6708^n$ \tcite{GolovachHP12}    \\
			Minimal Subset FVSs in chordal graphs &          $3$     &   $1.6667^n$   & NPR    \\
			Maximal $r$-colorable induced subgraphs of perfect graphs &          $r+1$     &   $(2-\frac{1}{r+1})^n$   & NPR    \\
			\bottomrule
		\end{tabular}
		
	}
	\caption{\label{fig:enumresults}Summary of known and new results for different combinatorial bounds. NPR means  that we are not aware of any previous results better than $2^n$. All bounds suppress factors polynomial in the universe size $n$.}
\end{table}

\medskip\noindent\textbf{Local Search versus Monotone Local Search.} 
One of the successful approaches for obtaining exact exponential-time algorithms for $d$-SAT  is based on sampling and local search. In his breakthrough paper  Sch{\"o}ning \cite{Schoening99-A} introduced the following simple and elegant approach: sample a random assignment and then do a local search in a Hamming ball of small radius around this assignment. With  the right choice of the parameter for the local search algorithm (the Hamming distance, in this case) it is possible to prove that with a reasonable amount of samples this algorithm decides the satisfiability of a given formula with good probability. The running time of Sch{\"o}ning's algorithm on formulas with $n$ variables is $\Oh((2-2/d)^n)$ and  it was shown by Moser and Scheder    \cite{MoserS11} how to derandomize it in almost the same running time, see also~\cite{Dantsin02deterministic}.  

While this method has been very successful for satisfiability, it is not clear how to apply this approach to other NP-complete problems, in particular to optimization problems, like finding a satisfying assignment of Hamming weight at most $k$ or finding a hitting set of size at most $k$. The reason why  Sch{\"o}ning's approach cannot be directly applied to optimization problems is that it is very difficult to get efficient local search algorithms for these problems. 

Consider for example \textsc{Min-Ones} $d$-\textsc{Sat}. 
If we select some assignment $a$ as a center of Hamming ball $B_r$ of radius $r$, there is a dramatic difference between searching for any satisfying assignment in $B_r$, and a satisfying assignment of Hamming weight at most $k$ in $B_r$. In the first case the local search problem can be solved in time $d^r \cdot n^{\Oh(1)}$. In the second case we do not know any better alternative to a brute-force search. Indeed, an algorithm with running time on the form $f(r) \cdot n^{\Oh(1)}$ for any function $f$ would imply that \FPT{} = \W[1]{}. This issue is not specific to \textsc{Min-Ones} $d$-\textsc{Sat}: it is known that the problem of  searching  a Hamming ball $B_r$ of radius $r$ is \W[1]-hard parameterized by $r$ for most natural optimization problems~\cite{FellowsFLRSV12}. 
 
Despite this obstacle, our approach is based on sampling an initial solution, and then performing a local search from that solution. The way we get around the hardness of local search is to make the local search problem easier, at the cost of reducing the success probability of the sampling step. Specifically, we only consider {\em monotone} local search, where we are not allowed to remove any elements from the solution, and only allowed to add at most $k$ new elements. Instead of searching a Hamming {\em ball} around the initial solution, we look for a solution in a Hamming {\em cone}. Monotone local search is equivalent to the extension problem, and it turns out that the extension problem can very often be reduced to the problem of finding a solution of size at most $k$. This allows us to use for our monotone local search the powerful toolbox developed for parameterized algorithms.  



\medskip\noindent\textbf{Our approach.}
Our algorithm is based on random sampling. Suppose we are looking for a solution $S$ of size $k$ in a universe $U$ of size $n$, and we have already found some set $X$ of size $t$ which we know is a subset of $S$. At this point, one option we have is to run the extension algorithm -- this would take time $c^{k-t}\cdot n^{\Oh(1)}$. Another option is to pick a {\em random} vertex $x$ from $U \setminus X$, add $x$ to $X$ and then proceed. We succeed if $x$ is in $S \setminus X$, so the probability of success is $(k-t)/(n-t)$. If we succeed in picking $x$ from $S \setminus X$ then $k-t$ drops by $1$, so running the extension algorithm on $X \cup \{x\}$ is a factor $c$ faster than running it on $X$. Therefore, as long as $(k-t)/(n-t) \geq 1/c$ it is better to keep sampling vertices and adding them to $X$. When $(k-t)/(n-t) < 1/c$ it is better to run the algorithm for the extension problem. This is the entire algorithm!

While the description of the algorithm is simple, the analysis is a bit more involved. At a first glance it is not at all obvious that a $100^k \cdot n^{\Oh(1)}$ time algorithm for the extension problem gives any advantage over trying all subsets of size $k$ in ${n \choose k}$ time. To see why our approach outperforms ${n \choose k}$, it is helpful to think of the brute-force algorithm as a randomized algorithm that picks a random subset of size $k$ by picking one vertex at a time and inserting it into the solution. The success probability of such an algorithm is 
\mymath{\frac{k}{n} \cdot \frac{k-1}{n-1} \cdot \frac{k-2}{n-2} \cdot \ldots \cdot \frac{2}{n-(k-2)} \cdot \frac{1}{n-(k-1)} = \frac{1}{{n \choose k}}.}
Notice that in the beginning of the random process the success probability of each step is high, but that it gets progressively worse, and that in the very end it is close to $1/n$. At some point we have picked $t$ vertices and $(k-t)/(n-t)$ drops below $1/c$. Here we run the extension algorithm, spending time $c^{k-t}$, instead of continuing with brute-force, which would take time
\mymath{{n-t \choose  k-t} = \frac{n-t}{k-t} \cdot \frac{n-t-1}{k-t-1} \cdot \ldots \cdot \frac{n-k+2}{2} \cdot \frac{n-k+1}{1}}
which is a product of $k-t$ larger and larger terms, with even the \longversion{first and }smallest term being greater than $c$. Thus we can conclude that any $c^k$ algorithm will give some improvement over $2^n$.

Notice that if the algorithm is looking for a set of size $k$ in a universe of size $n$, the number $t$ of vertices to sample before the algorithm should switch from picking more random vertices to running the extension algorithm can directly be deduced from $n, k$, and $c$. The algorithm picks a random set $X$ of size $t$, and runs the extension algorithm on $X$. We succeed if $X$ is a subset of a solution, hence the success probability is $p = {k \choose t}/{n \choose t}$. In order to get constant success probability, we run the algorithm $1/p$ times, taking time $c^{k-t}\cdot n^{\Oh(1)}$ for each run.

In order to derandomize the algorithm we \longversion{show that it is possible to }construct in time $(1/p) \cdot 2^{o(n)}$ a family ${\cal F}$ of sets of size $t$, such that $|{\cal F}| \leq (1/p) \cdot 2^{o(n)}$, and every set of size $k$ has a subset of size $t$ in ${\cal F}$. Thus, it suffices to \longversion{construct ${\cal F}$ and }run the extension algorithm on each set $X$ in ${\cal F}$. The construction of the family ${\cal F}$ lends ideas from Naor et al.~\cite{NaorSS95}, however their methods are not directly applicable to our setting. 

The main technical contribution of this paper is a non-trivial generalization of local-search based satisfiability algorithms to a wide class of optimization problems. Instead of covering the search space by Hamming balls, we cover it by Hamming cones and use a parameterized algorithm to search for a solution in each of the cones.
 

\section{Combining Random Sampling with FPT Algorithms}\label{sec:randalgo}

In this section we prove our main results, Theorems \ref{thm:main1}--\ref{thm:main4}, that will give new algorithms to find a set in $\mathcal{F}_I$ and to enumerate the sets in $\mathcal{F}_I$.
For many potential applications, the objective is to find a minimum-size set with certain properties, for example that the removal of this set of vertices yields an acyclic graph. This can easily be done using the algorithms resulting from Theorems \ref{thm:main1} and \ref{thm:main2} with only a polynomial overhead by using binary search over $k$, the size of the targeted set $S$, and specifying that $\mathcal{F}_I$ contains only sets of size at most $k$.

\subsection{Picking Random Subsets of the Solution.}

This subsection is devoted to the proof of Theorem \ref{thm:main1}.
The theorem will follow from the following lemma, which gives a new randomized algorithm for \phiext.
%
\begin{lemma}
\label{lemma:subext}
%
%
If there is a constant $c > 1$ and an algorithm for \phiext\ with running time $c^k \bitsize^{\Oh(1)}$, then there is a randomized algorithm for \phiext\ with running time $(2-\frac{1}{c})^{n-|X|}\bitsize^{\Oh(1)}$.
\end{lemma}

\begin{proof}
Let $\cal B$ be an algorithm for \phiext{} with running time $c^k\bitsize^{\Oh(1)}$. 
We now give another algorithm, $\cal A$, for the same problem. $\cal A$ is a randomized algorithm and consists of the following two steps for an input instance $(I,X,k')$ with $k'\le k$.
\begin{enumerate}
\shortversion{\setlength{\itemsep}{-2pt}}
\item Choose an integer $t \leq k'$ depending on $c$, $n$, $k'$ and $|X|$, and then select a random subset $Y$ of $U_I\setminus X$ of size $t$. The choice of $t$ will be discussed towards the end of the proof.
\item 
Run Algorithm $\cal B$ on the instance $(I, X \cup Y, k'-t)$ and return the 
answer. 
\end{enumerate}
This completes the description of \longversion{Algorithm }$\cal A$. Its running time is clearly upper bounded by $c^{k'-t} \bitsize^{\Oh(1)}$. 

If $\cal A$ returns \yes\ for $(I,X,k')$, this is because $\cal B$ returned \yes\ for $(I,X \cup Y, k' - t)$. In this case there exists a set $S \subseteq U_I \setminus (X \cup Y)$ of size at most $k' - t\le k-t$
such that $S\cup X \cup Y \in {\cal F}_I$. Thus, $Y \cup S$ witnesses that $(I,X,k)$ 
is indeed a \yesinstance.

Next we lower bound the probability that $\cal A$ returns \yes\ in case there exists a set $S \subseteq U_I \setminus X$ of size exactly $k'$ such that $X \cup S \in {\cal F}_I$.
The algorithm $\cal A$ picks a set $Y$ of size $t$ at random from $U_I \setminus X$. There are ${n-|X| \choose t}$ possible choices for $Y$. If $\cal A$ picks one of the ${k' \choose t}$ subsets of $S$ as $Y$ then $\cal A$ returns \yes. Thus, \longversion{given that there exists a set $S \subseteq U_I \setminus X$ of size $k'$ such that $X \cup S \in {\cal F}_I$, }we have that
\mymath{
\Pr\left[{\cal A} \mbox{ returns \yes} \right]  \geq  \Pr[Y \subseteq S]  = {{k' \choose t}}/{{n - |X| \choose t}}.
}


\noindent
Let $p(k') = {{k' \choose t}}/{{n - |X| \choose t}}$. For each $k'\in\{0,\dots,k\}$, our main algorithm runs ${\cal A}$ independently $1/p(k')$ times with parameter $k'$. The algorithm returns \yes\ \longversion{if any of the runs of}\shortversion{as soon as} ${\cal A}$ return \yes{}.
If $(I,X,k)$ is a \yes-instance,
then the main algorithm returns \yes\ with probability at least $\min_{k' \leq k} \{1-(1-p(k'))^{1/p(k')}\}\geq 1-\frac{1}{e} >\frac{1}{2}$. 
Next we upper bound the running time of the main algorithm, which is 
\mymath{
 \sum_{k' \leq k} \frac{1}{p(k')} \cdot c^{k'-t} \bitsize^{\Oh(1)} \leq  \max_{k' \leq k} \frac{{n - |X| \choose t}}{{k' \choose t}} \cdot c^{k'-t} \bitsize^{\Oh(1)} \leq \max_{k \leq n-|X|} \frac{{n - |X| \choose t}}{{k \choose t}} \cdot c^{k-t} \bitsize^{\Oh(1)}.
}

We are now ready to discuss the choice of $t$ in the algorithm ${\cal A}$. The algorithm ${\cal A}$ chooses the value for $t$ that gives the minimum value of $ \frac{{n - |X| \choose t}}{{k' \choose t}} \cdot c^{k'-t}$. Thus\longversion{, for fixed $n$ and $|X|$} the running time of the algorithm is\longversion{ upper bounded by}
\begin{align}\label{eqn:runtimeext}
 \max_{0 \leq k \leq n-|X|}\left\{  \min_{0 \leq t \leq k} \left\{  \frac{{n-|X| \choose t}}{{k \choose t}} c^{k-t} \bitsize^{\Oh(1)} \right\}\right\}.
\end{align}
We upper bound \longversion{the expression in \eqref{eqn:runtimeext}}\shortversion{this expression} by  
$\left(2-\frac{1}{c}\right)^{n-|X|}\bitsize^{\Oh(1)}$ in Lemma~\ref{lem:technical}\shortversion{ (see Appendix \ref{sec:omittedProofs})}. 
The running time of the algorithm is thus upper bounded by $\left(2-\frac{1}{c}\right)^{n-|X|}\bitsize^{\Oh(1)}$\longversion{, completing the proof}. 
\end{proof}

\begin{remark}
The proof of Lemma~\ref{lemma:subext} goes through just as well when $\cal B$ is a randomized algorithm. If $\cal B$ is deterministic or has one-sided error (possibly saying \no{} whereas it should say \yes{}), then the algorithm of Lemma~\ref{lemma:subext} also has one sided error. If $\cal B$ has two sided error, then the algorithm of  Lemma~\ref{lemma:subext} has two sided error as well.
\end{remark}

\longversion{
Now we give the technical lemma that was used to upper bound the running time of the algorithm described in Lemma~\ref{lemma:subext}. }
\newcommand{\lemmaTechnical}{
\begin{restatable}{lemma}{lemTechnical}
\label{lem:technical}
Let $c>1$ be a fixed constant, and \longversion{let }$n$ and $k \leq n$ be non-negative integers. Then,
\begin{eqnarray*}
\max_{0 \leq k \leq n}\left\{  \min_{0 \leq t \leq k} \left\{  \frac{{n \choose t}}{{k \choose t}} c^{k-t} \right\}\right\} \leq \left(2-\frac{1}{c}\right)^{n} n^{\Oh(1)}
\end{eqnarray*}
\end{restatable}
}
\longversion{\lemmaTechnical}
\newcommand{\proofLemTechnical}{%
\begin{proof}
Setting $\mu = \frac{k}{n}$ and $\alpha = \frac{t}{n}$, we have that
\[
\max_{0 \leq k \leq n}\left\{  \min_{0 \leq t \leq k} \left\{  \frac{{n \choose t}}{{k \choose t}} c^{k-t} \right\}\right\} = 
\max_{0\leq \mu \leq 1}\left\{  \min_{0\leq \alpha \leq \mu} \left\{  \frac{{n \choose \lceil \alpha n \rceil}}{ {\lceil \mu n \rceil \choose \lceil \alpha n \rceil}} c^{(\mu-\alpha)n}  \right\}\right\} \cdot \Oh(1)
\]
The right hand side of the equation above is upper bounded by picking a concrete value of $\alpha$ for every value of $\mu$, rather than minimizing over all $\alpha$. We set $\alpha=\max\left(0, \frac{1-c\mu}{1-c}\right)$. One can show that this choice of $\alpha$ minimizes the internal expression. In particular, this basically guarantees that $\frac{k-t}{n-t} = \frac{1}{c}$, and this is the natural threshold for when to stop sampling (see the discussion in the introduction).

First, consider the case where $\alpha=0$. The expression is upper bounded by $c^{n/c}$ since $1-c\mu\ge 0$.
To show that $\left(2-\frac{1}{c}\right)^n \ge c^{n/c}$, it suffices to show that $2-\frac{1}{c}-c^{1/c}\ge 0$.
But this is so because $2-\frac{1}{c}-c^{1/c}=0$ when $c=1$ and it is increasing with $c$ when $c>1$.

From now on, we assume $\alpha=\frac{1-c\mu}{1-c}>0$.
Thus,
\begin{align}\label{eqn:choseAlpha}
\max_{0\leq \mu \leq 1}\left\{  \min_{0\leq \alpha \leq \mu} \left\{  \frac{ {n \choose \lceil \alpha n \rceil}}{{\lceil \mu n \rceil \choose \lceil \alpha n \rceil}}c^{(\mu-\alpha)n}  \right\}\right\} 
\leq \max_{\begin{subxarray} 0\leq \mu \leq 1 \\ \alpha=\frac{1-c\mu}{1-c} \end{subxarray}} \left\{    \frac{ {n \choose \lceil \alpha n \rceil}}{{\lceil \mu n \rceil \choose \lceil \alpha n \rceil}}c^{(\mu-\alpha)n} \right\}.
\end{align}


We will also use the following well known bounds on binomial coefficients to simplify our expressions,  
\begin{align}\label{eqn:binBound}
\frac{1}{n^{\Oh(1)}}\left[\left(\frac{k}{n}\right)^{-\frac{k}{n}}\left(1 - \frac{k}{n}\right)^{\frac{k}{n} - 1}\right]^n \leq {n \choose k} \leq \left[\left(\frac{k}{n}\right)^{-\frac{k}{n}}\left(1 - \frac{k}{n}\right)^{\frac{k}{n} - 1}\right]^n.
\end{align}
Using the upper bound in Equation~\ref{eqn:choseAlpha} we obtain the following.
\begin{eqnarray}
\frac{ {n \choose \lceil \alpha n \rceil}}{{\lceil \mu n \rceil \choose \lceil \alpha n \rceil}}c^{(\mu-\alpha)n}
&\leq & \left(\frac{\alpha^{-\alpha} (1-\alpha)^{\alpha-1}}{(\frac{\alpha}{\mu})^{-\alpha}  (1-\frac{\alpha}{\mu})^{\alpha-\mu} } c^{(\mu-\alpha)}\right)^n  n^{\Oh(1)}\nonumber \\
& = &  \left(\frac{\alpha^{-\alpha} (1-\alpha)^{\alpha-1}} { \alpha^{-\alpha} \mu^\alpha (\mu-\alpha)^{\alpha -\mu}\mu^{(\mu-\alpha)}} c^{(\mu-\alpha)}\right)^n n^{\Oh(1)} \nonumber \\
& = & \left( (1-\alpha)^{\alpha-1}(\mu-\alpha)^{\mu -\alpha} \mu^{-\mu} c^{(\mu-\alpha)}\right)^n n^{\Oh(1)}
\label{eqn:runtimeA}
\end{eqnarray}
Substituting the value of $\alpha$ in Equation~\ref{eqn:runtimeA}, we get the following as the base of the exponent. 

\begin{eqnarray}
& &  \Big(\frac{c(1-\mu)}{c-1}\Big)^{ \frac{c(1-\mu)}{1-c}} \Big(\frac{1-\mu}{c-1}\Big)^{ \frac{1-\mu}{c-1}}
 \mu^{-\mu} c^{ \frac{1-\mu}{c-1}} \nonumber \\
 &=& \Big(\frac{c}{c-1}\Big)^{ \frac{c(1-\mu)}{1-c}+\frac{1-\mu}{c-1}} (1-\mu)^{ \frac{c(1-\mu)}{1-c}+\frac{1-\mu}{c-1}} ~\mu^{-\mu} \nonumber\\
 &= &  \Big(\frac{c}{c-1}\Big)^{\mu-1} (1-\mu)^{\mu-1} \mu^{-\mu} 
\label{eqn:runtimeB}
\end{eqnarray}
The last assertion in Equation~\ref{eqn:runtimeB} follows from the following simplification. 
\begin{eqnarray*}
 \frac{c(1-\mu)}{1-c}+\frac{1-\mu}{c-1}  =  \frac{c(1-\mu)-(1-\mu)}{1-c} =  \frac{(c-1)(1-\mu)}{1-c}=\mu-1. 
\label{eqn:runtimeC}
\end{eqnarray*}
To summarize the above discussion, we have upper bounded the expression in the statement of the lemma by
\begin{eqnarray*}
\max_{0 \leq k \leq n}\left\{  \min_{0 \leq t \leq k} \left\{  \frac{{n \choose t}}{{k \choose t}} c^{k-t} \right\}\right\} \leq \left
[\max_{0 \leq \mu \leq 1} f(\mu)\right]^n \cdot n^{O(1)},
\end{eqnarray*}
where $f(\mu)=\Big(\frac{c}{c-1}\Big)^{\mu-1} (1-\mu)^{\mu-1} \mu^{-\mu}$. 

We now turn to upper bounding the maximum of $f(\mu)$. Clearly, $f$ is continuous and differentiable on the interval $[0,1]$ and thus $f$ achieves its maximum at $\mu \in \{0,1\}$ or at a point where the derivative vanishes. Setting $\mu$ to $0$ and  to $1$, we get that $f(0)=\frac{c-1}{c}$ and $f(1)=1$, respectively. Next we differentiate $f$ with respect to $\mu$. The product rule for differentiation yields
\begin{eqnarray*}
f'(\mu)&=& f(\mu) \cdot  \left( \ln\left( \frac{c}{c-1}\right) + (\ln (1-\mu) +1) -1 -\ln \mu \right).
\end{eqnarray*}
%
%
%
Since $f(\mu) \neq 0$, we have that $f'(\mu) = 0$ if and only if 
$$ \ln\left( \frac{c}{c-1}\right) + \ln (1-\mu) -\ln \mu =0, $$
and the unique solution to this equation is $\mu = \frac{c}{2c-1}$.
%
%
Substituting this value of $\mu=\frac{c}{2c-1}$ in $f$ we get $f(\mu) = 2-\frac{1}{c}$.
%
%
%
%
Thus the maximum value $f$ can attain on $\mu \in [0,1]$ is the maximum of $1$, $1-\frac{1}{c}$ and $2-\frac{1}{c}$. Since, $c>1$ this implies that  $f(\mu) \leq 2-\frac{1}{c}$, completing the proof.
\end{proof}}
\longversion{\proofLemTechnical}


%
\noindent
By running the algorithm from Lemma~\ref{lemma:subext} with $X=\emptyset$ and for each value of $k\in \{0,\dots,n\}$, we obtain an algorithm for \phisub, and this proves Theorem \ref{thm:main1}.

\subsection{Derandomization}
In this subsection we prove Theorem \ref{thm:main2} by derandomizing the algorithm of Theorem~\ref{thm:main1}, at the cost of a subexponential factor in the running time. 
The key tool in our derandomization is a new pseudo-random object, which we call \sepfamwPlural{}, as well as an almost optimal (up to subexponential factors) construction of such objects. 
\begin{definition}
Let $U$ be a universe of size $n$ and let $0\leq q\leq p \leq n$. A family $\mathcal{C} \subseteq {U \choose q}$ is an \emph{\sepfam{n}{p}{q}}, if for every set $S \in {U \choose p}$, there exists a set $Y \in \cal C$ such that $Y \subseteq S$. 
\end{definition}

\noindent
Let $\kappa(n,p,q)= {{n \choose q}}/{{p \choose q}}$. In Section~\ref{sec:derandomization} (Theorem~\ref{thm:sepfamconstr}) we give a deterministic construction of an \sepfam{n}{p}{q}, $\cal C$, of size at most $\kappa(n,p,q) \cdot 2^{o(n)}$. The running time of the algorithm constructing  $\cal C$ is also upper bounded by $\kappa(n,p,q) \cdot 2^{o(n)}$.

%
The proof of Theorem~\ref{thm:main2} is now almost identical to the proof of Theorem~\ref{thm:main1}. However, in Lemma~\ref{lemma:subext} we replace the sampling step where the algorithm ${\cal A}$ picks a set $Y \subseteq U_I \setminus X$ of size $t$ at random, with a construction of an \sepfam{n-|X|}{k}{t} ${\cal C}$ using Theorem~\ref{thm:sepfamconstr}. Instead of  $\kappa(n-|X|,k,t) \cdot n^{O(1)}$ independent repetitions of the algorithm ${\cal A}$, the new algorithm loops over all $Y \in {\cal C}$. The correctness follows from the definition of \sepfamwPlural{}, while the running time analysis is identical to the analysis of Lemma~\ref{lemma:subext}.

\longversion{
\subsection{Extension to Permissive FPT Subroutines}
\label{subsec:permissive}
}
\newcommand{\secPermissive}{
For some of our applications, our results rely on algorithms for permissive variants of the \phiext\ problem.
Permissive problems were introduced in the context of local search algorithms \cite{MarxS11} and it has been shown that permissive variants can be fixed-parameter tractable even if the strict version is W[1]-hard and the optimization problem is NP-hard \cite{GaspersKOSS12}.

\defoptproblem{\pphiext}{An instance $I$, a set $X \subseteq U_I$, and an integer $k$.}
{If there is a subset $S \subseteq (U_I \setminus X)$ such that $S \cup X \in {\cal F}_I$ and $|S| \leq k$, then answer \yes;\newline
	else if $|{\cal F}_I|>0$, 
	then answer \yes\ or \no;\newline
	else answer \no.}

\noindent
We observe that any algorithm solving \phiext\ also solves \pphiext.
However, using an algorithm for \pphiext\ will only allow us to solve a decision variant of the \phisub\ problem, unless it also returns a certificate in case it answers \yes.

\defproblem{\dphisub}{An instance $I$}{Is $|{\cal F}_I|>0$?}

\noindent
The proof of Lemma \ref{lemma:subext} can easily be adapted to the \pphiext\ problem.

\begin{lemma}
	\label{lemma:subpext}
 If there exists a constant $c > 1$ and an algorithm for \pphiext\ with running time $c^k \bitsize^{\Oh(1)}$, then there exists a randomized algorithm for \pphiext\ with running time $(2-\frac{1}{c})^{n-|X|}\bitsize^{\Oh(1)}$.
\end{lemma}

\noindent
Now, any algorithm for \pphiext\ also solves \dphisub. If the algorithm for \pphiext\ also returns a certificate whenever it answers \yes, this also leads to an algorithm for \phisub.
%
%
Again, these algorithms can be derandomized at the cost of a factor $2^{o(n)}$ in the running time.

\begin{theorem}
	\label{thm:permissive}
	If there is an algorithm for \pphiext{} with running time $c^k\bitsize^{\Oh(1)}$ then there is an algorithm for \dphisub\ with running time $(2-\frac{1}{c})^{n+o(n)} \bitsize^{\Oh(1)}$.
	Moreover, if the algorithm for \pphiext{} computes a certificate whenever it answers \yes, then there is an algorithm for \phisub\ with running time $(2-\frac{1}{c})^{n+o(n)} \bitsize^{\Oh(1)}$.
\end{theorem}




}
\longversion{\secPermissive}

\subsection{Enumeration and Combinatorial Upper Bounds}

\longversion{In this subsection, we prove Theorems \ref{thm:main3} and \ref{thm:main4} on combinatorial upper bounds and enumeration algorithms.}
\shortversion{We now outline the proofs of Theorems \ref{thm:main3} and \ref{thm:main4} on combinatorial upper bounds and enumeration algorithms (for full proofs see Appendix \ref{sec:omittedProofs}).}
\shortversion{%
 For Theorem \ref{thm:main3}, consider the following random process:
 \begin{enumerate}
 	\setlength{\itemsep}{-2pt}
 	\item Choose an integer $t$ based on $c,$ $n$, and $k$, then randomly sample a subset $X$ of size $t$ from $U_I$. 
 	\item Uniformly at random pick a set $S$ from ${\mathcal F}_{I,X}^{k-t}$, and output $W = X \cup S$. In the special case where ${\mathcal F}_{I,X}^{k-t}$ is empty return the empty set.
 \end{enumerate}
 An analysis similar to the one in Lemma \ref{lemma:subext} shows that each set in the family $\mathcal{F}_I$ is selected with probability at least $(2 - \frac{1}{c})^{-n} \cdot n^{-\Oh(1)}$. This implies that there are at most $\left(2-\frac{1}{c}\right)^n n^{\Oh(1)}$ such sets.
}

\longversion{\thmComb*}
\newcommand{\proofThmComb}{
\begin{proof}
Let $I$ be an instance and $k \leq n$. We prove that the number of sets in ${\cal F}_I$ of size exactly $k$ is upper bounded by $\left(2-\frac{1}{c}\right)^n n^{\Oh(1)}$. Since $k$ is chosen arbitrarily the bound on $|{\cal F}_I|$ will follow. We describe below a random process that picks a set $W$ of size $k$ from ${\cal F}_I$ as follows.
\begin{enumerate}
\item Choose an integer $t$ based on $c$, $n$, and $k$, then randomly sample a subset $X$ of size  $t$ from $U_I$. 
\item Uniformly at random pick a set $S$ from ${\mathcal F}_{I,X}^{k-t}$, and output $W = X \cup S$. In the corner case where ${\mathcal F}_{I,X}^{k-t}$ is empty return the empty set.
\end{enumerate}
This completes the description of the process. 

For each set $Z \in {\cal F}_I$ of size exactly $k$, let $E_Z$ denote the event that the set $W$ output by the random process above is equal to $Z$. Now we lower bound the probability of the event $E_Z$.  We have the following lower bound.
\begin{eqnarray}
\Pr[E_Z]& = & \Pr[X \subseteq Z \wedge S = Z \setminus X] \nonumber \\
& = & \Pr[X \subseteq Z] \times \Pr[S =Z \setminus X~|~X \subseteq Z] \label{eqn:chosingfromasetC} \\
&=&
\frac{{k \choose t}}{{n \choose t}} \times \frac{1}{|{\mathcal F}_{I,X}^{k-t}|}
\nonumber
%
\end{eqnarray}
Since $\Phi$ is $c$-uniform we have that $|{\mathcal F}_{I,X}^{k-t}| \leq c^{k-t} n^{\Oh(1)}$, hence 
$$\Pr[E_Z] \geq \frac{{k \choose t}}{{n \choose t}}  c^{-(k-t)} n^{-\Oh(1)}.$$
We are now ready to discuss the choice of $t$ in the random process. The integer $t$ is chosen such that the above expression for $\Pr[E_Z]$ is maximized (or, in other words, it's reciprocal is minimized). By Lemma~\ref{lem:technical} we have that for every $k \leq n$ there exists a $t \leq k$ such that 
$$\frac{{k \choose t}}{{n \choose t}}  c^{-(k-t)} \geq (2 - \frac{1}{c})^{-n} \cdot n^{-\Oh(1)}.$$
Hence $\Pr[E_Z] \geq  (2 - \frac{1}{c})^{-n} \cdot n^{-\Oh(1)}$ for every $Z \in {\cal F}_I$ of size $k$. Since the events $E_Z$ are disjoint for all the different sets $Z \in {\cal F}_I$ we have that 
$$\sum_{\begin{subxarray} Z \in {\cal F}_I \\ |Z| = k \end{subxarray}} \Pr[E_{Z}] \leq 1.$$
This, together with the lower bound on $\Pr[E_Z]$ implies that the number of sets in ${\cal F}_I$ of size exactly $k$ is upper bounded by $\left(2-\frac{1}{c}\right)^n n^{\Oh(1)}$, completing the proof. 
%
%
\end{proof}%
}
\longversion{\proofThmComb}

\longversion{\noindent}%
If the implicit set system $\Phi$ is efficiently $c$-uniform then the proof of Theorem~\ref{thm:main3} can be made constructive by replacing the sampling step by a construction of an \sepfam{n}{k}{t} ${\cal C}$ using Theorem~\ref{thm:sepfamconstr}. For each $X \in {\cal C}$ the algorithm uses the fact that $\Phi$ is efficiently $c$-uniform to loop over all sets $S \in {\mathcal F}_{I,X}^{k-t}$ and output $X \cup S$ for each such $S$. Looping over  ${\cal C}$ instead of sampling $X$ incurs a $2^{o(n)}$ overhead in the running time of the algorithm. In order to avoid enumerating duplicates, we also store each set that we output in a trie and for each set that we output, we check first in linear time whether we have already output that set. \shortversion{This proves Theorem \ref{thm:main4}.}

\longversion{\thmEnum*}

\section{Efficient Construction of Set-Inclusion-Families}\label{sec:derandomization}

In this section we give the promised construction of \sepfamwPlural{}. We start by giving a construction of  \sepfamwPlural{} with good bounds on the size, but with a poor bound on the construction time. Recall that $\kappa(n,p,q)= {n \choose q} / {p \choose q}$.\shortversion{ The following lemma is proved in Appendix \ref{sec:omittedProofs}.}%
%
\begin{restatable}{lemma}{lemSlowBalancedUniversal}
\label{lem:slowBalancedUniversal}
There is an algorithm that given $n$, $p$ and $q$ outputs an \sepfam{n}{p}{q} $\cal C$ of size at most $\kappa(n,p,q) \cdot n^{\Oh(1)}$ in time $\Oh(3^n)$.
\end{restatable}
\newcommand{\proofLemSlowBalancedUniversal}{
\begin{proof}
We start by giving a randomized algorithm that with positive probability constructs an \sepfam{n}{p}{q}
 ${\cal C}$ with the claimed size. We will then discuss how to deterministically compute such a $\cal C$ within the required time bound. Set 
 $t =  \kappa(n,p,q)  \cdot (p+1)\log n$ 
 and construct the family ${\cal C} = \{C_1, \ldots, C_t\}$ by selecting each set $C_i$ independently and uniformly at random from ${U \choose q}$. 
 
By construction, the size of ${\cal C}$ is within the required bounds. We now argue that with positive probability ${\cal C}$ is indeed an \sepfam{n}{p}{q}. For a fixed set $A \in {U \choose p}$, and integer $i \leq t$, we consider the probability that $C_i \subseteq A$. This probability is
$1/{\kappa(n,p,q)}$.
Since each $C_i$ is chosen independently from the other sets in ${\cal C}$, the probability that {\em no} $C_i$ satisfies $C_i \subseteq A$ is
\begin{align*} \left(1 -  \frac{1}{\kappa(n,p,q)}\right)^t \leq e^{-(p+1)\log n} \leq \frac{1}{n^{p+1}}.\end{align*}
There are ${n \choose p}$ choices for $A \in {U \choose p}$, therefore the union bound yields that the probability that there exists an $A \in {U \choose p}$ such that no set $C_i \in {\cal C}$ satisfies $C_i \subseteq A$ is upper bounded by $\frac{1}{n^{p+1}} \cdot n^{p} = \frac{1}{n}$. 

To construct ${\cal C}$ within the stated running time proceed as follows. We construct an instance of {\sc Set Cover}, and then, using a known approximation algorithm for {\sc Set Cover}, we construct the desired family. An instance of {\sc Set Cover} consists of a universe $\cal U$ and a family $\mathscr{S}$ of subsets of $\cal U$. The objective is to find a minimum sized sub-collection $\mathscr{S}'\subseteq \mathscr{S}$ such that the union of elements of the sets in $\mathscr{S}'$ is $\cal U$. It is known that {\sc Set Cover} 
admits a polynomial time approximation algorithm with factor $\Oh(\log |{\cal U}|)$. For our problem, the 
elements of the universe $\cal U$ are $u_A$ for every $A\in {U \choose p}$. For every set $B\in {U \choose q}$, let $F_B$ consist of all the elements $u_A\in \cal U$ such that $B\subseteq A$. The set family $\mathscr S$ contains $F_B$ for each choice of $B\in {U \choose q}$. Given a sub-collection 
$\mathscr{S}'\subseteq \mathscr{S}$ we construct the family ${\cal C}(\mathscr{S}')$ by taking the sets $B\in {U \choose q}$ such that $F_B \in \mathscr{S}'$.  Clearly, any ${\cal C}(\mathscr{S}')$ corresponding to a sub-collection $\mathscr{S}'\subseteq \mathscr{S}$ covering $\cal U$ is a \sepfam{n}{p}{q}, and vice versa. 

Let {\sf OPT} denote the size of a minimum sized sub-collection $\mathscr{S}'\subseteq \mathscr{S}$ covering $\cal U$. 
We run the known $\Oh(\log |{\cal U}|)$-factor approximation algorithm on our instance and obtain a sub-collection $\mathscr{S}'\subseteq \mathscr{S}$ covering $\cal U$. Let ${\cal C}={\cal C}(\mathscr{S}')$. By discussions above we know that $\cal C$ is an \sepfam{n}{p}{q}. Clearly, the size of $\cal C$ is upper bounded by 
\[ |{\cal C}| \leq {\sf OPT} \cdot  \Oh(\log |{\cal U}|) \leq t \cdot  \Oh(\log |{\cal U}|) \leq \Oh(t  (\log n^p))\leq \kappa(n,p,q) \cdot n^{\Oh(1)}. \]
It is well known that one can implement the approximation algorithm for {\sc Set Cover} to run in time $\Oh(|{\cal U}|+\sum_{S\in \mathscr{S}} |S|)=\Oh({n \choose p}+{n \choose q} {n-q \choose p-q})=\Oh(3^n)$. 
%
This concludes the proof.
\end{proof}}
\longversion{\proofLemSlowBalancedUniversal}

\noindent
Next we will reduce the problem of finding an \sepfam{n}{p}{q} to the same problem, but with a much smaller value of $n$. To that end we will use a well-known construction of {\em pair-wise independent} families of functions. Let $U$ be a universe of size $n$ and $b$ be a positive integer. Let $\cal X$ be a collection of functions from $U$ to $[b]$. That is, each function $f$ in  $\cal X$ takes as input an element of $U$ and returns an integer from $1$ to $b$. The collection $\cal X$ is said to be {\em pair-wise independent} if, for every $i, j \in [b]$ and every $u, v \in U$ such that $u \neq v$ we have that
\mymath{\Pr_{f \in \cal X}[f(u) = i \wedge f(v) = j] = \frac{1}{b^2}.}
Observe that this implies that any pairwise independent family of functions from $U$ to $[b]$ with $|U| \geq 2$ also satisfies that for every $i \in [b]$ and $u \in U$ we have $\Pr_{f \in \cal X}[f(u) = i] = \frac{1}{b}$.
We will make use of the following known construction of pair-wise independent families.
\begin{proposition}[\cite{AlonBI86}]\label{prop:pairInd}
There is a polynomial time algorithm that given a universe $U$ and integer $b$ constructs a pair-wise independent family $\cal X$ of functions from $U$ to $[b]$. The size of $\cal X$ is $\Oh(n^2)$.
\end{proposition}

\noindent
Using Proposition~\ref{prop:pairInd} we can give a much faster construction of an \sepfam{n}{p}{q} than the one in Lemma~\ref{lem:slowBalancedUniversal} at the cost of a subexponential overhead in the size of the family.

\begin{theorem}\label{thm:sepfamconstr}
There is an algorithm that given $n$, $p$ and $q$ outputs an \sepfam{n}{p}{q} $\cal C$ of size at most $\kappa(n,p,q) \cdot 2^{o(n)}$ in time $\kappa(n,p,q) \cdot 2^{o(n)}$.
\end{theorem}

\begin{proof}
The construction sets $\beta = {q}/{p}$, selects a number $b = \lceil \log n \rceil$ of {\em buckets} and applies Proposition~\ref{prop:pairInd} to construct a pairwise independent family ${\cal X}$ of functions from $U$ to $[b]$. For each function $f \in {\cal X}$ and integer $i \in [b]$ we set $U_f^i = \{u \in U : f(u) = i\}$ and $n_f^i = |U_f^i|$. Call a function $f$ {\em good} if, for every $i \in [b]$ we have that $|n_f^i - {n}/{b}| \leq \sqrt{n} \cdot b$. For every good function $f \in {\cal X}$, every $i \in [b]$ and every integer $s \leq n_f^i$ we construct an \sepfam{n_f^i}{s}{\lceil \beta s \rceil} ${\cal C}_{f}^{(i,s)}$ using Lemma~\ref{lem:slowBalancedUniversal}. 
We now describe the family ${\cal C}$ output by the construction. Each set $Y \in {\cal C}$ is defined by 
\begin{enumerate}\setlength\itemsep{-4pt}
\item a good $f \in \cal X$, 
\item a sequence $p_1, \ldots ,p_b$ of integers such that $|p_i - \frac{p}{b}| \leq \sqrt{n} \cdot b$,
\item a sequence $Y_1, \ldots , Y_b$ of sets with $Y_i \in {\cal C}_{f}^{(i,p_i)}$,
\item a set $D \subseteq U$ of size at most $b$.
\end{enumerate}
The set $Y$ defined by the tuple $(f, p_1, \ldots ,p_b, Y_1, \ldots ,Y_b, D)$ is set to $Y = (\bigcup_{i \leq b} Y_i) \setminus D$.\longversion{ This concludes the construction.}

First we analyze the running time of the construction. Constructing the set ${\cal X}$ takes polynomial time by Proposition~\ref{prop:pairInd}. For each good $f$,  $i \in [b]$ and $s \leq n_f^i$, constructing ${\cal C}_{f}^{(i,s)}$ using Lemma~\ref{lem:slowBalancedUniversal} takes time $2^{o(n)}$ because $n_i \leq \frac{n}{b} + \sqrt{n} \cdot \log n = \Oh(\frac{n}{\log n})$. There are $\Oh(n^2)$ choices for $f$, at most $\Oh(\log n)$ choices for $i$ and $\Oh(\frac{n}{\log n})$ choices for $s$. Thus, the overall time of the construction is $2^{o(n)}$ plus the time to output ${\cal C}$. Outputting ${\cal C}$ can be done spending polynomial time for each set $Y \in {\cal C}$ by enumerating over all the tuples $(f, p_1, \ldots ,p_b, Y_1, \ldots ,Y_b, D)$. Thus, the running time of the construction is upper bounded by $2^{o(n)} + |{\cal C}| \cdot n^{\Oh(1)}$. It remains to upper bound $|{\cal C}|$.

The size of ${\cal C}$ is upper bounded by the number of tuples $(f, p_1, \ldots ,p_b, Y_1, \ldots ,Y_b, D)$. There are $\Oh(n^2)$ choices for $f,$ at most $n^{b}$ choices for $p_1, \ldots, p_b$ and $n^{\Oh(b)}$ choices for $D$. Thus, the number of tuples is upper bounded by $2^{o(n)}$ times the maximum number of choices for $Y_1 \ldots Y_b$ for any fixed choice of  $f,$ $p_1 \ldots p_b$ and $D$. For each $i$, we choose $Y_i$ from ${\cal C}_{f}^{(i,p_i)}$, so there are $\kappa(n_f^i, p_i, \lceil \beta p_i \rceil) \cdot n^{\Oh(1)}$ choices for $Y_i$. It follows that the total number of choices for $Y_1, \ldots ,Y_b$ is upper bounded by
\begin{align}\label{eqn:derandSizebound1}
\shortversion{\textstyle}
\prod_{i \leq b} \kappa(n_f^i, p_i, \lceil \beta p_i \rceil) \cdot n^{\Oh(1)} \leq 2^{o(n)} \cdot \prod_{i \leq b}
\shortversion{{{n_f^i \choose \lceil \beta p_i \rceil}}/{{p_i \choose \lceil \beta p_i \rceil}}.}
\longversion{\frac{{n_f^i \choose \lceil \beta p_i \rceil}}{{p_i \choose \lceil \beta p_i \rceil}}.}
\end{align}
Now, we have that
\begin{align}\label{eqn:upperprod}
\shortversion{\textstyle}
\prod_{i \leq b} {n_f^i \choose \lceil \beta p_i \rceil} \leq \prod_{i \leq b} \left[ {\lceil n/b \rceil \choose \lceil p/b \rceil} \cdot n^{\Oh(\sqrt{n}\log n)}\right] \leq {n \choose p} \cdot 2^{o(n)}
\end{align}
In the last transition we used that the number of ways to pick $b$ sets of size $\lceil p/b \rceil$, each from a universe of size $\lceil n/b \rceil$ is upper bounded by the number of ways to pick a set of size $b \cdot \lceil p/b \rceil$ from a universe of size $b \cdot \lceil n/b \rceil$. This in turn is upper bounded by ${n \choose p} \cdot 2^{o(n)}$. Furthermore, 
\begin{align}\label{eqn:lowerprod}
\shortversion{\textstyle}
\prod_{i \leq b} {p^i \choose \lceil \beta p_i \rceil} & \geq \prod_{i \leq b} \left[{\lceil p/b \rceil \choose \lceil \beta (p/b) \rceil} \cdot n^{-\Oh(\sqrt{n}\log n)}\right]
\longversion{\nonumber \\
 & \geq \left[\left(\beta^{-\beta}(1-\beta)^{\beta - 1}\right)^{(p/b)} \cdot n^{-\Oh(1)}\right]^b \cdot 2^{-o(n)} \\&
}
\geq {p \choose q} \cdot 2^{-o(n)}\longversion{\nonumber}
\end{align}
\longversion{%
	Here the two last transitions use Equation~\ref{eqn:binBound}.%
}
Inserting the bounds from \eqref{eqn:upperprod} and~\eqref{eqn:lowerprod} into \eqref{eqn:derandSizebound1} yields that the total number of choices for $Y_1, \ldots ,Y_b$ is upper bounded by $\kappa(n,p,q) \cdot 2^{o(n)}$ and thus, $|{\cal C}| \leq \kappa(n,p,q) \cdot 2^{o(n)}$ as well.

All that remains is to argue that ${\cal C}$ is in fact an \sepfam{n}{p}{q}. Towards this, consider any subset $S$ of $U$ of size exactly $p$. For any fixed $i \in [b]$, consider the process of picking a random function $f$ from ${\cal X}$. We are interested in the random variables $|U_f^i|$ and $|U_f^i \cap S|$. Using indicator variables for each element in $U$ it is easy to show that 
\mymath{
 \underset{f \in {\cal X}}{ {\rm E}}\left[|U_f^i|\right] = \frac{n}{b} \mbox{ and } \underset{f \in {\cal X}}{ {\rm E}}\left[|U_f^i \cap S|\right] = \frac{p}{b}.
}
Furthermore, ${\cal X}$ is pairwise independent, and therefore the covariance of any pair of indicator variables is $0$. Thus, $\underset{f \in {\cal X}}{ {\rm Var}}\left[|U_f^i|\right] \leq n$ and $\underset{f \in {\cal X}}{ {\rm Var}}\left[|U_f^i \cap S|\right] \leq n$. By Chebyshev's inequality it follows that
\mymath{
 \Pr\left[\big||U_f^i| - \frac{n}{b}\big| \geq \sqrt{n} \cdot b\right] \leq \frac{1}{b^2} \mbox{ and } \Pr\left[\big||U_f^i \cap S| - \frac{p}{b}\big| \geq \sqrt{n} \cdot b\right] \leq \frac{1}{b^2}.
}
Consider now the probability that at least one of the variables $|U_i|$ or $|U_i \cap S|$ deviates from its expectation by at least $\sqrt{n} \cdot b$. Combining the above inequalities with the union bound taken over all $i \in [b]$ yields that this probability is upper bounded by $2b \cdot \frac{1}{b^2} \leq \frac{2}{b}$. Since $b = \log n > 2$ we have that with non-zero probability, {\em all} the random variables $|U_1|, \ldots , |U_b|$ and $|U_1 \cap S|, \ldots , |U_b \cap S|$ are within $\sqrt{n} \cdot b$ of their respective means. Thus there exists a function $f \in {\cal X}$ such that for every $i \in [b]$ we have
\mymath{
 \big||U_f^i| - \frac{n}{b}\big| \leq \sqrt{n} \cdot b \mbox{ and } \big||U_f^i \cap S| - \frac{p}{b}\big| \leq \sqrt{n} \cdot b.
}
\longversion{In the remainder of the proof let}\shortversion{Let} $f$ be such a function in ${\cal X}$.

The choice of $f$ implies that $f$ is a good function. For each $i \leq b$, let $S_i = |U_f^i \cap S|$ and $p_i = |S_i|$. Again, by the choice of $f$ we have that $|p_i - \frac{p}{b}| \leq \sqrt{n} \cdot b$. Since ${\cal C}_{f}^{(i,p_i)}$ is an \sepfam{n_f^i}{p_i}{\lceil \beta p_i \rceil}, there exists a set $Y_i \in {\cal C}_{f}^{(i,p_i)}$ such that $Y_i \subseteq S_i$ and $|Y_i| = \lceil \beta p_i \rceil$. For each $i \in [b]$ select such a $Y_i$ from ${\cal C}_{f}^{(i,p_i)}$. Finally let $D$ be any subset of $\bigcup_{i \leq b} Y_i$ of size $\sum_{i \leq b} |Y_i| - q$. Note that $|Y_i| \leq \beta p_i + 1$, thus $\sum_{i \leq b} |Y_i| - q \leq b$, so $|D| \leq b$. Consider finally the tuple $(f, p_1 \ldots p_b, Y_1, \ldots Y_b, D)$. We have just proved that this tuple satisfies all of the conditions for giving rise to a set $Y = \bigcup_{i \leq b} Y_i \setminus D$ in ${\cal C}$. However, $Y_i \subseteq S_i$ for all $i$, so $Y \subseteq S$, proving that ${\cal C}$ is a \sepfam{n}{p}{q}.
\end{proof}

\section{Conclusion and Discusison}\label{sec:concl}

In this paper we have shown that for many subset problems, an algorithm that finds a solution of size $k$ in time $c^kn^{\Oh(1)}$ directly implies an algorithm with running time $\Oh((2-\frac{1}{c})^{n+o(n)})$. We also show that often, an upper bound of $c^kn^{\Oh(1)}$ on the number of sets of size at most $k$ in a family ${\cal F}$ can yield an upper bound of $\Oh((2-\frac{1}{c})^{n+o(n)})$ on the size of ${\cal F}$.
Our results reveal an exciting new connection between parameterized algorithms and exponential-time algorithms.
All of our algorithms have a randomized and a deterministic variant. The only down-side of using the deterministic algorithm rather than the randomized one is a $2^{o(n)}$ multiplicative factor in the running time, and an additional $2^{o(n)}$ space requirement. It is possible to reduce the space overhead to a much smaller (but still super-polynomial) term, however this would make the presentation considerably more involved.

For the enumeration algorithm of Theorem~\ref{thm:main4}, it is well worth noting that the algorithm only uses subexponential space if the algorithm is allowed to output the same set multiple times. If duplicates are not allowed the algorithm needs exponential space in order to store a trie 
of the sets that have already been output.
Another approach is to use an output-sensitive algorithm. For example, there is a polynomial-delay polynomial-space algorithm enumerating all feedback vertex sets in a tournament \cite{GaspersM13}, and its running time is $O(1.6667^n)$ by our combinatorial upper bound.

Our analysis also reveals that in order to obtain a $(2-\epsilon)^n$ time algorithm with $\epsilon > 1$ for a subset problem, it is sufficient to get a $\Oh(c^k{n-|X| \choose k}^{1-\delta})$ algorithm for any constant $c$ and $\delta > 0$ for the extension problem. This might be a promising route for obtaining better exact exponential-time algorithms for problems that currently do not have single-exponential-time parameterized algorithms. For example, it would be interesting to see whether it is possible to improve on Razgon's $\Oh(1.9977^n)$ time algorithm~\cite{Razgon07} for {\sc Directed Feedback Vertex Set} by designing a $\Oh(c^k{n-|X| \choose k}^{1-\delta})$ time algorithm.


\longversion{
 \medskip\noindent\textbf{Acknowledgements.} 
 Many thanks to Russell Impagliazzo and  Meirav Zehavi for insightful discussions.
 The research leading to these results has received funding from the European Research Council under the European Union's Seventh Framework Programme (FP/2007-2013) / ERC Grant Agreements n. 267959 and no. 306992.
 NICTA is funded by the Australian Government through the Department of Communications and the Australian Research Council (ARC) through the ICT Centre of Excellence Program.
 Serge Gaspers is the recipient of an ARC Future Fellowship (project number FT140100048) and acknowledges support under the ARC's Discovery Projects funding scheme (project number DP150101134). Lokshtanov is supported by the Beating Hardness by Pre-processing grant under the recruitment programme of the of Bergen Research Foundation.
}

\shortversion{\newpage}

\bibliographystyle{siam}
\bibliography{book_pc,references}

\begin{thebibliography}{10}

\bibitem{AgrawalLKS16}
{\sc A.~Agrawal, D.~Lokshtanov, S.~Kolay, and S.~Saurabh}, {\em A faster {FPT}
  algorithm and a smaller kernel for block graph vertex deletion}, in
  Proceedings of the 12th Latin American Theoretical Informatics Symposium
  (LATIN 2016), Lecture Notes in Computer Science, Springer, to appear.
\newblock Available as arXiv CoRR abs/1510.08154.

\bibitem{AlonBI86}
{\sc N.~Alon, L.~Babai, and A.~Itai}, {\em A fast and simple randomized
  parallel algorithm for the maximal independent set problem}, J. Algorithms, 7
  (1986), pp.~567--583.

\bibitem{BasteFKS15}
{\sc J.~Baste, L.~Faria, S.~Klein, and I.~Sau}, {\em Parameterized complexity
  dichotomy for $(r, \ell)$-vertex deletion}, Tech. Report abs/1504.05515,
  arXiv CoRR, 2015.

\bibitem{BliznetsFPV13}
{\sc I.~Bliznets, F.~V. Fomin, M.~Pilipczuk, and Y.~Villanger}, {\em Largest
  chordal and interval subgraphs faster than $2^n$}, in Proceedings of the 21st
  Annual European Symposium on Algorithms (ESA), vol.~8125 of Lecture Notes in
  Comput. Sci., Springer, 2013, pp.~193--204.

\bibitem{BoralCKP14}
{\sc A.~Boral, M.~Cygan, T.~Kociumaka, and M.~Pilipczuk}, {\em A fast branching
  algorithm for cluster vertex deletion}, in Proceedings of the 9th
  International Computer Science Symposium in Russia (CSR), vol.~8476 of
  Lecture Notes in Computer Science, Springer, 2014, pp.~111--124.

\bibitem{BrandstadtLS99}
{\sc A.~Brandst{\"a}dt, V.~B. Le, and J.~P. Spinrad}, {\em Graph Classes: A
  Survey}, Discrete Mathematics and Applications, SIAM, 1999.

\bibitem{Cao15}
{\sc Y.~Cao}, {\em Unit interval editing is fixed-parameter tractable}, in
  Proceedings of the 42nd International Colloquium of Automata, Languages and
  Programming (ICALP), vol.~9134 of Lecture Notes in Comput. Sci., Springer,
  2015, pp.~306--317.

\bibitem{Cao8kinterval}
{\sc Y.~Cao}, {\em Linear recognition of almost interval graphs}, in
  Proceedings of the 27th Annual ACM-SIAM Symposium on Discrete Algorithms
  (SODA 2016), to appear.
\newblock Available as arXiv CoRR abs/1403.1515.

\bibitem{CaoCL15}
{\sc Y.~Cao, J.~Chen, and Y.~Liu}, {\em On feedback vertex set: New measure and
  new structures}, Algorithmica, 73 (2015), pp.~63--86.

\bibitem{ChenFLLV08}
{\sc J.~Chen, F.~V. Fomin, Y.~Liu, S.~Lu, and Y.~Villanger}, {\em Improved
  algorithms for feedback vertex set problems}, J. Computer and System
  Sciences, 74 (2008), pp.~1188--1198.

\bibitem{CochefertCGK16}
{\sc M.~Cochefert, J.-F. Couturier, S.~Gaspers, and D.~Kratsch}, {\em Faster
  algorithms to enumerate hypergraph transversals}, in Proceedings of the 12th
  Latin American Theoretical Informatics Symposium (LATIN 2016), Lecture Notes
  in Computer Science, Springer, to appear.
\newblock Available as arXiv CoRR abs/1510.05093.

\bibitem{cygan2015parameterized}
{\sc M.~Cygan, F.~V. Fomin, L.~Kowalik, D.~Lokshtanov, D.~Marx, M.~Pilipczuk,
  M.~Pilipczuk, and S.~Saurabh}, {\em Parameterized Algorithms}, Springer,
  2015.

\bibitem{cut-and-count}
{\sc M.~Cygan, J.~Nederlof, M.~Pilipczuk, M.~Pilipczuk, J.~M.~M. van Rooij, and
  J.~O. Wojtaszczyk}, {\em Solving connectivity problems parameterized by
  treewidth in single exponential time}, in Proceedings of the 52nd Annual
  Symposium on Foundations of Computer Science (FOCS), IEEE, 2011,
  pp.~150--159.

\bibitem{Dantsin02deterministic}
{\sc E.~Dantsin, A.~Goerdt, E.~A. Hirsch, R.~Kannan, J.~Kleinberg,
  C.~Papadimitriou, P.~Raghavan, and U.~Sch{\"o}ning}, {\em A deterministic
  {$(2-2/(k+1))\sp n$} algorithm for {$k$}-{SAT} based on local search},
  Theoretical Computer Science, 289 (2002), pp.~69--83.

\bibitem{DehneFLRS07}
{\sc F.~K. H.~A. Dehne, M.~R. Fellows, M.~A. Langston, F.~A. Rosamond, and
  K.~Stevens}, {\em An {$O(2^{O(k)}n^{3})$} {FPT} algorithm for the undirected
  feedback vertex set problem}, Theory of Computing Systems, 41 (2007),
  pp.~479--492.

\bibitem{FellowsFLRSV12}
{\sc M.~R. Fellows, F.~V. Fomin, D.~Lokshtanov, F.~A. Rosamond, S.~Saurabh, and
  Y.~Villanger}, {\em Local search: Is brute-force avoidable?}, J. Comput.
  Syst. Sci., 78 (2012), pp.~707--719.

\bibitem{FominGKLS10}
{\sc F.~V. Fomin, S.~Gaspers, D.~Kratsch, M.~Liedloff, and S.~Saurabh}, {\em
  Iterative compression and exact algorithms}, Theor. Comput. Sci., 411 (2010),
  pp.~1045--1053.

\bibitem{FominGPR08-On}
{\sc F.~V. Fomin, S.~Gaspers, A.~V. Pyatkin, and I.~Razgon}, {\em On the
  minimum feedback vertex set problem: Exact and enumeration algorithms},
  Algorithmica, 52 (2008), pp.~293--307.

\bibitem{FominHKPV14}
{\sc F.~V. Fomin, P.~Heggernes, D.~Kratsch, C.~Papadopoulos, and Y.~Villanger},
  {\em Enumerating minimal subset feedback vertex sets}, Algorithmica, 69
  (2014), pp.~216--231.

\bibitem{Fomin:2010mo}
{\sc F.~V. Fomin and D.~Kratsch}, {\em Exact Exponential Algorithms}, Springer,
  2010.
\newblock An EATCS Series: Texts in Theoretical Computer Science.

\bibitem{FominTV15}
{\sc F.~V. Fomin, I.~Todinca, and Y.~Villanger}, {\em Large induced subgraphs
  via triangulations and {CMSO}}, {SIAM} J. Comput., 44 (2015), pp.~54--87.

\bibitem{GaspersKOSS12}
{\sc S.~Gaspers, E.~J. Kim, S.~Ordyniak, S.~Saurabh, and S.~Szeider}, {\em
  Don't be strict in local search!}, in Proceedings of the Twenty-Sixth {AAAI}
  Conference on Artificial Intelligence (AAAI), {AAAI} Press, 2012.

\bibitem{GaspersM13}
{\sc S.~Gaspers and M.~Mnich}, {\em Feedback vertex sets in tournaments},
  Journal of Graph Theory, 72 (2013), pp.~72--89.

\bibitem{GolovachHP12}
{\sc P.~A. Golovach, P.~van~'t Hof, and D.~Paulusma}, {\em Obtaining planarity
  by contracting few edges}, Theoretical Computer Science, 476 (2013),
  pp.~38--46.

\bibitem{GuoGHNW06}
{\sc J.~Guo, J.~Gramm, F.~H{\"u}ffner, R.~Niedermeier, and S.~Wernicke}, {\em
  Compression-based fixed-parameter algorithms for feedback vertex set and edge
  bipartization}, J. Computer and System Sciences, 72 (2006), pp.~1386--1396.

\bibitem{KanjLLTXXYZZZ14}
{\sc I.~A. Kanj, G.~Lin, T.~Liu, W.~Tong, G.~Xia, J.~Xu, B.~Yang, F.~Zhang,
  P.~Zhang, and B.~Zhu}, {\em Algorithms for cut problems on trees}, in
  Proceedings of the 8th International Conference on Combinatorial Optimization
  and Applications (COCOA), vol.~8881 of Lecture Notes in Computer Science,
  Springer, 2014, pp.~283--298.

\bibitem{Kante0KP15}
{\sc M.~M. Kant{\'{e}}, E.~J. Kim, O.~Kwon, and C.~Paul}, {\em An {FPT}
  algorithm and a polynomial kernel for linear rankwidth-1 vertex deletion}, in
  Proceedings of the 10th International Symposium on Parameterized and Exact
  Computation ({IPEC} 2015), vol.~43 of LIPIcs, Schloss Dagstuhl -
  Leibniz-Zentrum fuer Informatik, 2015, pp.~138--150.

\bibitem{KociumakaP13}
{\sc T.~Kociumaka and M.~Pilipczuk}, {\em Faster deterministic feedback vertex
  set}, Inf. Process. Lett., 114 (2014), pp.~556--560.

\bibitem{KolayP15}
{\sc S.~Kolay and F.~Panolan}, {\em Parameterized algorithms for deletion to
  $(r, \ell)$-graphs}, in Proceedings of FSTTCS 2015, to appear.
\newblock Available as arXiv CoRR abs/1504.08120.

\bibitem{KumarL16}
{\sc M.~Kumar and D.~Lokshtanov}, {\em Faster exact and parameterized algorithm
  for feedback vertex set in tournaments}, in Proceedings of the 33rd
  International Symposium on Theoretical Aspects of Computer Science (STACS
  2016), to appear.
\newblock Available as arXiv CoRR abs/1510.07676.

\bibitem{abs-1007-1166}
{\sc K.~Kutzkov and D.~Scheder}, {\em Using {CSP} to improve deterministic
  {3-SAT}}, CoRR, abs/1007.1166 (2010).

\bibitem{MarxS11}
{\sc D.~Marx and I.~Schlotter}, {\em Stable assignment with couples:
  Parameterized complexity and local search}, Discrete Optimization, 8 (2011),
  pp.~25--40.

\bibitem{Moon:1971xz}
{\sc J.~W. Moon}, {\em On maximal transitive subtournaments}, Proc. Edinburgh
  Math. Soc. (2), 17 (1971), pp.~345--349.

\bibitem{MoserS11}
{\sc R.~A. Moser and D.~Scheder}, {\em A full derandomization of
  {S}ch{\"{o}}ning's {$k$-SAT} algorithm}, in Proceedings of the 43rd Annual
  ACM Symposium on Theory of Computing (STOC), {ACM}, 2011, pp.~245--252.

\bibitem{NaorSS95}
{\sc M.~Naor, L.~J. Schulman, and A.~Srinivasan}, {\em Splitters and
  near-optimal derandomization}, in Proceedings of the 36th Annual Symposium on
  Foundations of Computer Science (FOCS), IEEE, 1995, pp.~182--191.

\bibitem{Razgon06}
{\sc I.~Razgon}, {\em Exact computation of maximum induced forest}, in
  Proceedings of the 10th Scandinavian Workshop on Algorithm Theory (SWAT
  2006), vol.~4059 of Lecture Notes in Comput. Sci., Berlin, 2006, Springer,
  pp.~160--171.

\bibitem{Razgon07}
\leavevmode\vrule height 2pt depth -1.6pt width 23pt, {\em Computing minimum
  directed feedback vertex set in o(1.9977\({}^{\mbox{n}}\))}, in Proceedings
  of the 10th Italian Conference on Theoretical Computer Science (ICTCS), 2007,
  pp.~70--81.

\bibitem{Schoening99-A}
{\sc U.~Sch{\"o}ning}, {\em A probabilistic algorithm for {$k$}-{SAT} and
  constraint satisfaction problems}, in 40th Annual Symposium on Foundations of
  Computer Science (New York, 1999), IEEE Computer Soc., Los Alamitos, CA,
  1999, pp.~410--414.

\bibitem{ShachnaiZ15}
{\sc H.~Shachnai and M.~Zehavi}, {\em A multivariate approach for weighted
  {FPT} algorithms}, in Proceedigs of the 23rd Annual European Symposium on
  Algorithms ({ESA}), vol.~9294 of Lecture Notes in Computer Science, Springer,
  2015, pp.~965--976.

\bibitem{HofV13}
{\sc P.~van~'t Hof and Y.~Villanger}, {\em Proper interval vertex deletion},
  Algorithmica, 65 (2013), pp.~845--867.

\bibitem{MagnusPhD07}
{\sc M.~Wahlstr{\"o}m}, {\em Algorithms, measures and upper bounds for
  satisfiability and related problems}, PhD thesis, Link\"oping University,
  Sweden, 2007.

\bibitem{Wahlstrom14}
\leavevmode\vrule height 2pt depth -1.6pt width 23pt, {\em Half-integrality,
  {LP}-branching and {FPT} algorithms}, in Proceedings of the 24th Annual
  ACM-SIAM Symposium on Discrete Algorithms (SODA), SIAM, 2014, pp.~1762--1781.

\bibitem{XiaoN13}
{\sc M.~Xiao and H.~Nagamochi}, {\em Exact algorithms for maximum independent
  set}, in Proceedings of the 24th International Symposium on Algorithms and
  Computation (ISAAC), vol.~8283 of Lecture Notes in Computer Science,
  Springer, 2013, pp.~328--338.
\newblock See also arXiv CoRR abs/1312.6260.

\bibitem{YatoS03}
{\sc T.~Yato and T.~Seta}, {\em Complexity and completeness of finding another
  solution and its application to puzzles}, IEICE Transactions on Fundamentals
  of Electronics, Communications and Computer Sciences, E86-A (2003),
  pp.~1052--1060.

\end{thebibliography}

 \newpage
 \newpage
\appendix

\section{Problem definitions}
\label{sec:problems}

We list the definitions of the problems considered in this paper.

\smallskip
\defparproblem{\sc Feedback Vertex Set}{An undirected graph $G$ and a positive integer $k$.}{$k$}{Does there exist a subset $S\subseteq V(G)$ of size at most $k$ such that $G-S$ is acyclic?}

\defparproblem{\sc Weighted Feedback Vertex Set}{An undirected graph $G$, a positive integer $k$, a weight function $w:V(G)\rightarrow \mathbb{N}$, and a positive integer $W$.}{$k$}{Is there a set $S\subseteq V(G)$ of size at most $k$ and weight at most $W$ such that $G-S$ is acyclic?}

\defparproblem{\sc Subset Feedback Vertex Set}{An undirected graph $G$, a vertex subset $T\subseteq V(G)$, and a positive integer $k$.}{$k$}{Does there exist a subset $S\subseteq V(G)$ of size at most $k$ such that $G-S$ has no cycle that contains a vertex from $T$?}
\smallskip
%

Let $\Gamma$ be a finite group with identity element $1_\Gamma$.
A $\Gamma$-labeled graph is a graph $G=(V,E)$ with a labeling $\lambda: E \rightarrow \Gamma$ such that $\lambda(u,v)\lambda(v,u)=1_\Gamma$ for every edge $uv\in E$.
For a cycle $C=(v_1,\dots,v_r,v_1)$, define $\lambda(C) = \lambda(v_1,v_2)\cdot \dots \cdot \lambda(v_r,v_1)$.

\smallskip
\defparproblem{\sc  Group Feedback Vertex Set}{A group $\Gamma$, a $\Gamma$-labelled graph $(G,\lambda)$, and a positive integer $k$.}{$k$}{Does there exist a subset $S\subseteq V(G)$ of size at most $k$ such that every cycle $C$ in $G-S$ has $\lambda(C)=1_\Gamma$?}
\smallskip

\defparproblem{\textsc{Node Unique Label Cover}}{An undirected graph $G=(V,E)$, a finite alphabet $\Sigma$, an integer $k$, and for each edge $e\in E$ and each of its endpoints $v$ a permutation $\psi_{e,v}$ of $\Sigma$ such that if $e=xy$ then $\psi_{e,x} = \psi_{e,v}^{-1}$}{$|\Sigma|+k$}{Is there a vertex subset $S\subset V$ of size at most $k$ and a function $\Psi: V\setminus S \rightarrow \Sigma$ such that for every edge $uv\in E(G-S)$ we have $(\Psi(u),\Psi(v))\in \psi_{uv,u}$?}

For  fixed integers $r,\ell \geq 0$, a graph $G$ is called an {\em \abGraph} if the vertex set $V(G)$ can be partitioned into $r$ independent sets and $\ell$ cliques.

\smallskip
\defparproblem{{\abpartization}}{A graph $G$ and a positive integer $k$}{$k$}{Is there a vertex subset $S\subseteq V(G)$ of size at most $k$ such that  
	$G-S$ is an \abGraph?} 

\smallskip
Several special cases of this problem are well known and have been widely studied. For example, $(2,0)$- and  $(1,1)$-graphs correspond to bipartite graphs and split graphs respectively.
We note that \abpartization\ can be solved in $O(1.1996^{(r+\ell)\cdot n})$ by taking $r$ copies of the input graph, $\ell$ copies of its complement, making all the copies of a same vertex into a clique and computing a maximum independent set of this graph using the algorithm from \cite{XiaoN13}.
This is faster than $O(2^n)$ when $r+\ell\le 3$.
We improve on this algorithm for $r,\ell\le 2$ and $r+\ell\ge 3$.

For the definition of graph classes, including interval graphs, proper interval graphs, block graphs, cluster graphs, we refer to \cite{BrandstadtLS99}.

\defparproblem{\sc Proper Interval Vertex Deletion}{An undirected graph $G$ and a positive integer $k$.}{$k$}{Does there exist a subset $S\subseteq V(G)$ of size at most $k$ such that $G-S$ is a proper interval graph?}

\defparproblem{\sc Interval Vertex Deletion}{An undirected graph $G$ and a positive integer $k$.}{$k$}{Does there exist a subset $S\subseteq V(G)$ of size at most $k$ such that $G-S$ is an interval graph?}

\defparproblem{\sc  Block Graph Vertex Deletion}{An undirected graph $G$ and a positive integer $k$.}{$k$}{Does there exist a subset $S\subseteq V(G)$ of size at most $k$ such that $G-S$ is a block graph?}

\defparproblem{\sc  Cluster Vertex Deletion}{An undirected graph $G$ and a positive integer $k$.}{$k$}{Does there exist a subset $S\subseteq V(G)$ of size at most $k$ such that $G-S$ is a cluster graph?}

\defparproblem{\sc Thread Graph Vertex Deletion} {An undirected graph $G$ and a positive integer $k$.}{$k$}{Does there exist a subset $S\subseteq V(G)$ of size at most $k$ such that $G-S$ is  of linear rank-width one?}

\defparproblem{\sc Multicut on Trees}{A tree $T$ and a set $\mathcal{R} = \{ \{s_1,t_1\}, \ldots, \{s_r,t_r\}\}$ of pairs of vertices of $T$ called terminals, and a positive integer $k$.}{$k$}{Does there exist a subset 
	$S\subseteq E(T)$ of size at most $k$ whose removal disconnects each $s_i$ from $t_i$, $i\in [r]$?}

\defparproblem{\sc $d$-Hitting Set}{A family $\mathscr S$ of subsets of size at most $d$ of a universe ${\cal U}$  and a positive integer $k$.}{$k$}{Does there exist a subset $S\subseteq \cal U$ of size at most $k$ such that  $F$ is a hitting set for 
	$\mathscr S$?}

\defparproblem{\sc Weighted $d$-Hitting Set}{A family $\mathscr S$ of subsets of size at most $d$ of a universe ${\cal U}$, a weight function $w:\mathcal{U}\rightarrow \mathbb{N}$, and positive integers $k$ and $W$.}{$k$}{Does there exist a subset $S\subseteq \cal U$ of size at most $k$ and weight at most $W$ such that  $F$ is a hitting set for $\mathscr S$?}

\defparproblem{\sc Min-Ones $d$-Sat}{A propositional formula $F$ in conjunctive normal form where each clause has at most $d$ literals and an integer $k$.}{$k$}{Does $F$ have a satisfying assignment with Hamming weight at most $k$?}

\defparproblem{\sc Weighted $d$-Sat}{A propositional formula $F$ in conjunctive normal form where each clause has at most $d$ literals, a weight function $w:var(F)\rightarrow \mathbb{Z}$, and integers $k$ and $W$.}{$k$}{Is there a set $S\subseteq var(F)$ of size at most $k$ and weight at most $W$ such that $F$ is satisfied by the assignment that sets the variables in $S$ to $1$ and all other variables to $0$?}

\defparproblem{\sc Tournament Feedback Vertex Set}{A tournament $G$ and a positive integer $k$.}{$k$}{Does there exist a subset $S\subseteq V(G)$ of size at most $k$ such that $G-S$ is a transitive 
tournament?}

\defparproblem{\sc  Split Vertex Deletion}{An undirected graph $G$ and a positive integer $k$.}{$k$}{Does there exist a subset $S\subseteq V(G)$ of size at most $k$ such that $G-S$ is a split graph?}

\defparproblem{\sc  Cograph Vertex Deletion}{An undirected graph $G$ and a positive integer $k$.}{$k$}{Does there exist a subset $S\subseteq V(G)$ of size at most $k$ such that $G-S$ is a cograph?}
%

\defparproblem{\sc Directed Feedback Vertex Set}{A directed graph $G$ and a positive integer $k$.}{$k$}{Does there exist a subset $S\subseteq V(G)$ of size at most $k$ such that $G-S$ is directed acyclic graph?}

%
%

\shortversion{
\section{Omitted Proofs}
\label{sec:omittedProofs}

In this section of the appendix, we provide the proofs that are omitted from the main text.

\lemmaTechnical
\proofLemTechnical

\thmComb*
\proofThmComb

\lemSlowBalancedUniversal*
\proofLemSlowBalancedUniversal

\section{Extension to Permissive FPT Subroutines}
\label{subsec:permissive}
\secPermissive
}

\end{document}